\documentclass[12pt]{article}

\usepackage{amssymb}

\usepackage[utf8]{inputenc}
\usepackage[T1]{fontenc}
\usepackage[english]{babel}
\usepackage[dvipsnames]{xcolor}
\usepackage{amsmath}
\usepackage{amsfonts}
\usepackage{amssymb}
\usepackage{bm}
\usepackage{xfrac}
\usepackage[round]{natbib}
\usepackage{enumitem}
\usepackage{booktabs}
\usepackage{xfrac}
\usepackage{csquotes}
\usepackage{tabularx}
\usepackage{graphicx}
\usepackage{multirow}
\usepackage{colortbl} 
\usepackage{pgfplotstable}
\pgfplotsset{compat=1.17}
\usepackage{booktabs}
\usepackage{float}
\usepackage{algorithm}
\usepackage{algpseudocode} 
\algnewcommand\algorithmicnot{\textbf{not}}
\algdef{SE}[IF]{IfNot}{EndIf}[1]{\algorithmicif\ \algorithmicnot\ #1\ \algorithmicthen}{\algorithmicend\ \algorithmicif}%

\usepackage{tikz}
\usetikzlibrary{shapes.misc}
\usepackage{ifsym} 

\usepackage{changepage}
\usepackage{caption}
\usepackage{authblk} 
\usepackage[pdfencoding=auto, psdextra, hidelinks]{hyperref}
\allowdisplaybreaks 
\usepackage{amsthm}
\theoremstyle{definition}
\newtheorem{example}{Example}
\newtheorem{definition}{Definition}

\theoremstyle{plain}

\newtheorem{proposition}{Proposition}
\newtheorem{corollary}{Corollary}

\newcommand{\Sol}{\mathcal{S}}
\newcommand{\M}{\mathcal{M}}
\newcommand{\N}{\mathcal{N}}
\newcommand{\Y}{\mathcal{Y}}
\newcommand{\A}{\mathcal{A}}

\newcommand{\T}{\mathcal{T}}
\newcommand{\x}{\bm{x}}
\newcommand{\y}{\bm{y}}

\usepackage{pifont}
\newcommand{\no}{\color{Red}{\ding{55}}}
\newcommand{\yes}{\color{Green}{\checkmark}}

\newcommand{\CC}{C\nolinebreak\hspace{-.05em}\raisebox{.4ex}{\tiny\bf +}\nolinebreak\hspace{-.10em}\raisebox{.4ex}{\tiny\bf +}}
\def\CC{{C\nolinebreak[4]\hspace{-.05em}\raisebox{.4ex}{\tiny\bf ++}}}

\title{Fair integer programming under dichotomous and cardinal preferences}
\author[a,b,*]{Tom Demeulemeester}
\author[b,c]{Dries Goossens}
\author[a,d]{Ben Hermans}
\author[a]{Roel Leus}

\affil[a]{Research Center for Operations Research \& Statistics, KU Leuven, Belgium}
\affil[b]{Department of Business Informatics and Operations Management, Ghent University, Belgium}
\affil[c]{Corelab CVAMO, FlandersMake@UGent, Belgium}
\affil[d]{ORTEC, Belgium}
\affil[*]{Corresponding author: \tt{tom.demeulemeester@kuleuven.be}}
\date{}
\begin{document}
\maketitle




\vspace*{-1cm}
\begin{abstract}
\noindent One cannot make truly fair decisions using integer linear programs unless one controls the selection probabilities of the (possibly many) optimal solutions. For this purpose, we propose a unified framework when binary decision variables represent agents with \emph{dichotomous} preferences, who only care about whether they are selected in the final solution. We develop several general-purpose algorithms to fairly select optimal solutions, for example, by maximizing the Nash product or the minimum selection probability, or by using a random ordering of the agents as a selection criterion (Random Serial Dictatorship). We also discuss in detail how to extend the proposed methods when agents have \emph{cardinal} preferences. As such, we embed the \enquote{black-box} procedure of solving an integer linear program into a framework that is explainable from start to finish. Lastly, we evaluate the proposed methods on two specific applications, namely kidney exchange (dichotomous preferences), and the scheduling problem of minimizing total tardiness on a single machine (cardinal preferences). We find that while the methods maximizing the Nash product or the minimum selection probability outperform the other methods on the evaluated welfare criteria, methods such as Random Serial Dictatorship perform reasonably well in computation times that are similar to those of finding a single optimal solution. 
\end{abstract}




\section{Introduction}
\label{sec:intro}
When solving an integer linear program (ILP), solvers traditionally return one of the possibly many optimal solutions in a deterministic way \citep[e.g.,][]{cplex2021deterministic, gurobi2023deterministic}. This might not be desirable in practical applications where the implications of the selected solution are of great importance to the agents involved. Consider, for example, the following zero-one knapsack instance.
\begin{align*}
    &\max \quad \quad 2x_1 + x_2 + x_3 + x_4\\
    &\text{  s.t.\ } \quad \quad \: 2x_1 + x_2 + x_3 + x_4 \leq 3\\
    &\quad \quad \quad \quad \; x_1, x_2, x_3, x_4 \in \{0,1\}
\end{align*}
This simple instance could represent a wide range of practical problems. Imagine, for example, that a school has three remaining spots while five students want to enroll at that school. Decision variable~$x_1$ represents two of those students who are twins, and only want to be selected together, while the other decision variables represent individual students.

There are four optimal solutions for this instance: $\{x_1, x_2\}$, $\{x_1, x_3\}$, $\{x_1, x_4\}$, and $\{x_2, x_3, x_4\}$. One possible way to fairly choose between these solutions is to select solution~$\{x_2, x_3, x_4\}$ with a probability of~40\%, and each of the other solutions with a probability of~20\%. In this way, each student is selected with an equal probability of~60\%. Nevertheless, commercial solvers, such as Gurobi and CPLEX, will always return solution $\{x_2, x_3, x_4\}$ for their default parameter settings, regardless of the order in which the variables are input into the solver.\footnote{We tested this for the default settings of Gurobi 9.1.1 and CPLEX 12.9, both implemented in \CC.} This means that Gurobi and CPLEX will never select the twins in our example, for no clear reason. Using a solver without being aware of this issue may therefore result in an unfair treatment of some of the agents involved. Moreover, this small example shows that one cannot claim to make a fair decision using an ILP unless one controls the selection procedure of the optimal solutions.

To overcome this undesirable behaviour, we will discuss methods to control the selection probabilities of the optimal solutions of ILPs with multiple optimal solutions, in order to improve both fairness and transparency in decision-making processes that use ILPs. Since it is well-known that enumerating all optimal solutions of an ILP formulation is computationally challenging in general, we will pay special attention to methods that do not require a full enumeration. We will refer to the problem of controlling the selection probabilities for the optimal solutions of ILPs as \emph{fair integer programming}. In Sections \ref{sec:partitioning}-\ref{sec:axioms}, we study the fair integer programming problem for ILPs with binary decision variables, each representing an agent with \emph{dichotomous} preferences. Dichotomous preferences can be used to model simple settings where a yes/no decision should be made for each agent to decide whether they are selected in the final solution, or to model more complex settings where agents only care whether a certain criterion is satisfied by the final solution, e.g., an agent is happy if and only if \enquote{their} set is covered by the final solution, or if and only if the agent's submodular utility function reaches a certain treshold. The general class of ILPs that we study can be used to model various network problems, knapsack, facility location, scheduling, matching, kidney exchange, etc. In Section~\ref{sec:extensions}, we will discuss in detail how the proposed methods can be extended to settings in which the agents have \emph{cardinal} preferences, i.e., a utility value for each of the possible outcomes.
    
Important to note is that our analysis takes place after the decision-maker has implicitly described a set of optimal solutions that are all equally desirable in her opinion. In other words, we assume the set of optimal solutions to satisfy all inequity and group fairness criteria that the decision-maker deems relevant for a specific application \citep[e.g.,][]{karsu2015inequity}. The input to our problem is then a formulation that describes this set of optimal solutions, but it is irrelevant for our methods which (linear) constraints it contains or whether or not it is the last step of a hierarchical optimization process. One should note that symmetry breaking and dominance rules, which are typically used to lower computation times of ILPs by reducing the number of optimal solutions, should be adopted with care. When using a dominance rule, for example, which exploits the fact that a non-empty subset of the optimal solutions satisfies a certain property, the resulting formulation is no longer guaranteed to describe all optimal solutions to the original instance.\footnote{These dominance rules are referred to as existential- and sufficient-property based dominance rules by \cite{jouglet2011dominance}. Adopting dominance rules that describe a necessary condition for optimality (universal-property based dominance rules), however, will not harm the outcome of the fair integer programming problem.} Because we optimize various fairness criteria over the convex hull of the optimal solutions, this might result in selection probabilities over the optimal solutions that are sub-optimal with respect to the chosen fairness criterion. 

Our main contributions are the following. First, to the best of our knowledge, we propose the first framework that is not application-specific to tackle fair integer programming under dichotomous preferences. In contrast to the recent results by \cite{flanigan2021fair} and \cite{st2022adaptation}, we propose general-purpose algorithms for a highly general class of ILPs to construct selection probabilities over the optimal solutions that satisfy various fairness criteria (Sections~\ref{sec:partitioning} and~\ref{sec:distributions}).

Second, the generality of our framework allows us to embed the fair integer programming problem into the rich literature on probabilistic social choice and cooperative bargaining. The connection with probabilistic social choice enables us to study axiomatic properties of the proposed methods (Section~\ref{sec:axioms}). Additionally, we exploit the connection with cooperative bargaining to extend our framework to agents with cardinal preferences, thus initiating the study of this problem setting in the literature (Section~\ref{sec:extensions}).

Third, we evaluate the proposed methods for two specific applications, one with dichotomous and one with cardinal preferences (Section~\ref{sec:comput}). We find that the performance of the proposed methods on the evaluated welfare criteria is rather application-specific. Interestingly, our findings show that the Random Serial Dictatorship rule that we introduce for fair integer programming, which uses a random ordering of the agents as a selection criterion, is an intuitive mechanism that performs reasonably well on the evaluated welfare criteria in solution times that are similar to those of finding a single optimal solution.


\section{Definitions}
\label{sec:preliminaries}
We start by defining the general class of integer linear programs for which we will study the selection procedure of an optimal solution when the agents have dichotomous preferences. Define a set $\mathcal{A}$ of $n\in \mathbb{N}$ agents with corresponding binary decision variables~$\bm{x}\in \{0,1\}^n$, a matrix~$\bm{A}\in \mathbb{R}^{m \times n}$, and vectors~$\bm{v}\in \mathbb{R}^n$ and $\bm{c}\in\mathbb{R}^m$, with $m\in\mathbb{N}$. Additionally, consider a vector of $k\in \mathbb{N}$ auxiliary integer decision variables $\bm{y} \in \mathbb{Z}^k$, with corresponding parameters~$\bm{w}\in \mathbb{R}^k$ and $\bm{B} \in \mathbb{R}^{m\times k}$.\footnote{Most results in our paper continue to hold when the $\y$-variables are continuous, rather than integer. While the set of optimal solutions might be infinite in both cases, causing the uniform distribution in Section~\ref{subsec:uniform}, for example, to be ill-defined, the projection of the set of optimal solutions onto the $\x$-variables will still be finite.} Consider the following integer linear program.
\begin{align*}
		\max \;\;& \bm{v}^\intercal \bm{x} + \bm{w}^\intercal \bm{y}\\
		\text{s.t.} \;\;& \bm{Ax} + \bm{By} \leq \bm{c}\\
		& \bm{x} \in \{0,1\}^n, \y\in\mathbb{Z}^k
\end{align*} 
Denote the set of all ILPs of the above form by~$\Xi$. For each instance~$\xi \in \Xi$, a binary decision has to be made for each of the agents in~$\mathcal{A}$, and we say that an agent~$i\in \mathcal{A}$ is \emph{selected} in a given solution~$\bm{x}$ if $x_i = 1$.

Let~$\mathcal{S}(\xi) = \{(\bm{x^1},\bm{y^1}), (\bm{x^2},\bm{y^2}),\ldots, (\bm{x^{|\Sol(\xi)|}},\bm{y^{|\Sol(\xi)|}})\}$ be the set of all optimal solutions of an ILP~$\xi \in \Xi$. Moreover, denote the objective value of the solutions in~$\mathcal{S}(\xi)$ by~$z^*(\xi)$. We are mainly interested in the projection of $\Sol(\xi)$ onto the $\x$-variables, which we denote by $\Sol_{\x}(\xi)$. We will simply refer to~$\mathcal{S}(\xi)$ by~$\mathcal{S}$ when the ILP~$\xi \in \Xi$ is clear from the context, and the same holds for~$z^*$, $\Sol_{\x}$, and other related notations that will be introduced further on. Moreover, we will denote the convex hull of $\Sol_{\x}$ by $\text{Conv}(\Sol_{\x})$.

Based on~$\mathcal{S}$, we can partition the set of agents~$\mathcal{A}$ in the following disjoint subsets:
\begin{enumerate}[label=(\roman*)]
    \itemsep 0em
    \item $\mathcal{Y} = \{i \in \mathcal{A}: \forall s \in \mathcal{S}: x^s_i = 1\}$;
    \item $\mathcal{N} = \{i \in \mathcal{A}: \forall s \in \mathcal{S}:x^s_i = 0\}$;
    \item $\mathcal{M} = \{i \in \mathcal{A}: \exists s,t \in \mathcal{S}: x^s_i \neq x^t_i\}$.
\end{enumerate}
The set~$\mathcal{Y}$, resp.\ $\mathcal{N}$, consists of the agents that are always, resp.\ never, selected, while the set~$\mathcal{M}$ contains the agents that are selected in some, but not in all of the optimal solutions in~$\mathcal{S}$. Unless there exists a solution that selects all agents in~$\M$, any deterministic selection of one of the solutions~$(\bm{x^s},\bm{y^s}) \in \mathcal{S}$ will clearly disadvantage at least one agent~$i\in\mathcal{M}$ for which~$x^s_i =0$. A fair treatment of the agents in~$\mathcal{M}$ therefore requires randomization. Given a set of optimal solutions~$\mathcal{S}$, a \emph{lottery}~$\bm{\lambda}= (\lambda_s)_{s=1}^{|\mathcal{S}|}$ is a probability distribution over~$\mathcal{S}$, with $\sum_{s=1}^{|\mathcal{S}|} \lambda_s = 1$ and $\lambda_s \geq 0$ for all~$s\in\{1,\ldots, |\mathcal{S}|\}$. Denote the set of all lotteries for a set of optimal solutions~$\mathcal{S}$ by~$\Delta(\mathcal{S})$. 

In general, decision-makers mostly care about the selection probabilities of the agents for ILPs in~$\Xi$, rather than about the selection probabilities of the optimal solutions. A \emph{distribution} is a vector $\bm{d}\in[0,1]^n$, corresponding to the selection probabilities of the agents in~$\mathcal{A}$ for an ILP in~$\Xi$. We assume probability~$d_i$ to be the canonical utility of agent~$i\in\mathcal{A}$ \citep[similarly to, e.g.,][]{aziz2019fair}. Clearly, not all such vectors can be obtained through lotteries over the optimal solutions, as illustrated in Example~\ref{ex:intro} below. The following definition formalizes this idea.

\begin{definition}
\label{def:implementable} Given an ILP~$\xi \in \Xi$, a distribution~$\bm{d} \in [0,1]^n$ is \emph{realizable} over the corresponding set of optimal solutions $\mathcal{S} = \{(\bm{x^1},\bm{y^1}),\ldots,(\bm{x^{|\mathcal{S}|}},\bm{y^{|\mathcal{S}|}})\}$ if there exists a lottery~$\bm{\lambda}\in\Delta(\mathcal{S})$ such that
\begin{equation}
    \label{eq:describes}
    \bm{d} = \sum_{s=1}^{|\mathcal{S}|} \lambda_s\bm{x^s}.  
\end{equation}
\end{definition}
A lottery~$\bm{\lambda} \in \Delta(\mathcal{S})$ that satisfies Equation~(\ref{eq:describes}) is said to \emph{realize}~$\bm{d}$, and we denote the set of all lotteries that realize a distribution~$\bm{d}$ by $\Delta_{\bm{d}}(\mathcal{S}) \subseteq \Delta(\mathcal{S})$. Note that Definition~\ref{def:implementable} is equivalent to saying that, given an ILP in~$\Xi$, a distribution~$\bm{d}$ is realizable over a set of optimal solutions~$\mathcal{S}$ if it lies in the convex hull of the~$\bm{x}$-variables of the solutions in~$\mathcal{S}$.

Lastly, a \emph{distribution rule} is a function~$f:\Xi\to[0,1]^n$ that maps each integer linear program $\xi \in\Xi$ to a distribution~$f(\xi) \in [0,1]^n$. Because of the one-to-one mapping between an ILP~$\xi$ and its set of optimal solutions~$\mathcal{S}(\xi)$, we will use $f(\xi)$ and $f(\Sol)$ interchangeably in the remainder of this paper. 

The following example illustrates the introduced terminology.

\begin{example}
\label{ex:intro}
Consider the following zero-one knapsack instance with four agents, and a capacity of 6.
\begin{align*}
    &\max \quad \quad 4x_1 + 3x_2 + x_3 + x_4\\
    &\text{  s.t.\ } \quad \quad \: 4x_1 + 2.5x_2 + 2.5x_3 + 2.5x_4 \leq 6\\
    &\quad \quad \quad \quad \; x_1, x_2, x_3, x_4 \in \{0,1\}
\end{align*}
The optimal objective value for this instance equals~$z^*=4$, and there are three optimal solutions, namely $\mathcal{S}=\{(1,0,0,0), (0,1,1,0), (0,1,0,1)\}$. Because all agents appear in some, but not in all of the optimal solutions in~$\mathcal{S}$, all agents belong to~$\mathcal{M}$, and~$\mathcal{Y} = \mathcal{N} = \varnothing$. A possible lottery~$\bm{\lambda^U}$ is to select each of the optimal solutions in~$\mathcal{S}$ with an equal probability, i.e., $\bm{\lambda^U} = \{\frac{1}{3},\frac{1}{3},\frac{1}{3}\}$. The corresponding distribution~$\bm{d^U}$ is equal to~$(\frac{1}{3},\frac{2}{3},\frac{1}{3},\frac{1}{3})$, which means that agent~2 is selected with a probability of~$\frac{2}{3}$ by lottery~$\bm{\lambda^U}$, while all other agents are selected with a probability of~$\frac{1}{3}$. Now consider the distribution~$\bm{d^E} = (\eta, \eta, \eta, \eta)$, which select all agents in~$\mathcal{M}$ with an equal probability~$\eta$. Then~$\bm{d^E}$ is not realizable in this instance for any value~$\eta\in[0,1]$, because no lottery~$\bm{\lambda} = (\lambda_1, \lambda_2, \lambda_3)$ satisfies $\lambda_1 + \lambda_2 + \lambda_3 = 1$, while also satisfying $\lambda_1 = \lambda_2 = \lambda_3 = \eta$, and~$\lambda_2+\lambda_3=\eta$.

\end{example}

\section{Related work}
\label{sec:related}
Although the body of literature that deals with fairness and transparency in algorithmic decision-making is vast and rapidly expanding, very few papers explicitly discuss the problem of how to select one of the optimal solutions of some general ILP in a fair and transparent way. Nevertheless, because of the generality of the problem at hand, various fields of research cover topics that are closely related to it, and in the remainder of this section we aim to provide a concise overview of the relevant literature.

\subsection{Fair integer programming for specific problems}
We will first discuss two specific problem settings, which are both special cases of the general setting studied in this paper, in which the fair integer programming problem has been studied. First, \citet{flanigan2021fair} and \cite{flanigan2021transparent} study the selection probabilities of optimal solutions for \emph{sortition}, which is the problem of randomly selecting a panel of representatives from the population to decide on policy questions. The constraints in their model are simply quota stating lower and upper bounds for various subsets of the population (e.g., female, older than 65). While \cite{flanigan2021fair} study the distribution rules that we discuss in Sections~\ref{subsec:uniform}-\ref{subsec:custom_selection_criteria}, and propose a column generation framework for them, \cite{flanigan2021transparent} study how to implement these distribution rules as a uniform lottery over a set of~$m$ panels. 

A second specific problem setting for which the selection probabilities of optimal solutions have been studied in the literature is \emph{kidney exchange}. In kidney exchange, patients who suffer from kidney failure and who have an incompatible kidney donor, are matched to the incompatible donor of another patient such that the matched donors' kidneys can be transplanted. While various formulations to model the kidney exchange problem as an ILP exist  \citep[e.g.,][]{abraham2007clearing,roth2007efficient,dickerson2016position}, the objective typically consists of maximizing the number of transplants. As pointed out by \cite{farnadi2021individual} and \cite{carvalho2023theoretical}, however, there may be many solutions maximizing the number of transplants. \cite{roth2005pairwise} introduce an egalitarian mechanism, which equalizes the individual probabilities of receiving a transplant as much as possible, and which outputs a lottery over the maximum-size matchings for pairwise exchanges (no exchange cycles of size three or larger). \cite{li2014egalitarian} show that Roth et al.'s egalitarian solution can be computed efficiently. Alternatively, \cite{farnadi2021individual} propose and evaluate three different methods to enumerate all maximum-size matchings for kidney exchange problems with longer exchange cycles, and then discuss how to 
optimize two families of probability distributions over the optimal solutions. Moreover, \cite{st2022adaptation} propose a column generation procedure for the rules discussed in Section \ref{subsec:custom_selection_criteria}, and for the maximin rule (which is only the first step of the leximin rule discussed in Section~\ref{subsec:leximin}). The column generation procedures that we propose in Section~\ref{sec:distributions} allow the decision-maker to have more control over which of the solutions can be included in the resulting lottery. \cite{st2022adaptation} balance maximizing some fairness measure and the quality of the solutions in the lottery by optimizing a weighted product of both, possibly resulting in undesirable solutions in the support of the resulting lotteries. By including all solutions that are at most a fraction~$\epsilon<1$ worse than the optimal solution (Section~\ref{sec:near_opt}), our framework can be used to find a distribution minimizing individual fairness without the empty solution in the support, thus solving one of their open questions.

Further, a recent stream of papers study, for settings where decisions have to be made repeatedly, how to improve fairness over time \citep[e.g.,][]{bampis2018fair,lackner2020perpetual,lodi2022fairness,elkind2022fairness}.

\subsection{Cooperative bargaining and probabilistic social choice}
There are two more general problem settings, each with their own terminology and solution concepts, in which our problem can be embedded. First, in \emph{cooperative bargaining}, a set of two or more participants is faced with a \emph{feasible region} in the utility space, which is generally assumed to be non-empty, convex, closed, and bounded. If the participants can reach a unanimous agreement on a point in this feasible region, then each of the participants receives the corresponding utility. If unanimity cannot be reached, a given \textit{disagreement outcome} is the result. We refer the reader to \cite{Roth1979axiomatic}, \cite{thomson1994cooperative}, and \cite{peters2013axiomatic} for a detailed overview of results. In our case, the feasible set corresponds to the convex hull of the optimal solutions of the ILP under consideration, while the disagreement outcome is the origin for all agents in~$\M$, who are selected in some, but not in all of the optimal solutions. In Section~\ref{subsubsec:bargaining}, we elaborate on the link between our setting and cooperative bargaining when agents have cardinal, rather than dichotomous, preferences.

Second, in \emph{probabilistic social choice}, all agents report their (ordinal) preferences over a set of outcomes. The goal is then to select a lottery over the set of outcomes in order to satisfy certain desirable criteria. Well-studied probabilistic social choice functions are \emph{Random (Serial) Dictatorship} \citep{gibbard1977manipulation}, and \emph{maximal lotteries} \citep{fishburn1984probabilistic,brandl2016consistent}. Our problem can be stated in the terminology of probabilistic social choice theory by letting each of the optimal solutions correspond to one of the outcomes. 

Following the corresponding literature in probabilistic social choice under dichotomous preferences \citep[e.g.,][]{bogomolnaia2004random,bogomolnaia2005collective,aziz2019fair}, we will make the assumption that an agent's canonical utility for a lottery over the optimal solutions of an ILP is simply the expected value of the binary variable associated to them, namely the probability with which she is selected by the lottery. In the context of cooperative bargaining, so-called \textit{binary lottery games} have been experimentally studied, for example, by \cite{roth1979game}, \cite{roth1981sociological}, \cite{roth1982role}, and \cite{murnighan1988risk}.

The main difference between fair integer programming and the literature on cooperative bargaining and probabilistic social choice is the way in which the feasible region or the set of possible outcomes is expressed. An underlying assumption in cooperative bargaining and probabilistic social choice is that the set of possible outcomes is given \emph{explicitly}, by describing the non-empty, convex, closed, and bounded feasible region (cooperative bargaining), or by listing all possible outcomes (probabilistic social choice). In our setting, however, the set of possible outcomes is described \emph{implicitly} as the set of optimal solutions to an integer programming formulation. This implies that, in general, it is $\mathcal{NP}$-hard, and thus computationally challenging, to already obtain one of the optimal solutions \citep[e.g.,][]{karp1972reducibility}. Moreover, returning the set of all optimal solutions to an ILP belongs to complexity class $\#\mathcal{P}$, which is the analogue to $\mathcal{NP}$, but defined for counting problems instead of for decision problems \citep{valiant1979complexity_permanent,valiant1979complexity_enumeration,danna2007generating}. As a result, we put strong emphasis on methods that obtain a maximally fair lottery over the optimal solutions of an integer linear program, for various fairness metrics, without having to generate the set of all optimal solutions.

\subsection{Other related work}
We conclude this section by giving a concise overview of other streams of literature related to our setting. First, \cite{danna2007generating} and \cite{serra2020compact} illustrate that many ILPs with binary decision variables have multiple, and possibly many, optimal solutions for MIPLIB instances, and \cite{farnadi2021individual} and \cite{carvalho2023theoretical} illustrate this for kidney exchange instances. Moreover, \cite{serra2020compact} discuss how to represent (near)-optimal solutions in weighted decision diagrams, which can be easily queried. An alternative to generating all optimal solutions is to sample one of the optimal solutions. One possible solution method for the more general problem of random sampling in convex bodies are \textit{(geometric) random walks}, and we refer to reader to \cite{vempala2005geometric} for an overview.

With respect to the fairness of the returned solution, \cite{chen2022combining} provide a recent overview of the related problem of selecting a utility vector from a set of feasible utility vectors in order to maximize a given social welfare function, without allowing for randomization. Moreover, \cite{bertsimas2011price} introduce the \emph{price of fairness} concept, which quantifies the relative loss in utility between the utility-maximizing solution and the maximally fair solution. \cite{michorzewski2020price} extend their results when agents have dichotomous preferences, and \cite{dickerson2014price} and \cite{mcelfresh2018balancing} study the price of fairness in kidney exchange.

\section{Partitioning the agents}
\label{sec:partitioning}
When the set of optimal solutions~$\mathcal{S}$ cannot be fully enumerated, the partitioning of the set of agents~$\mathcal{A}$ into the disjoint subsets~$\mathcal{Y}$, $\mathcal{N}$, and~$\mathcal{M}$ is crucial to obtain fair selection probabilities for the agents involved. Indeed, when this is not done systematically, an agent that actually belongs to~$\mathcal{M}$ might be falsely considered to belong to~$\mathcal{Y}$ or to~$\mathcal{N}$, thus being unrightfully advantaged, resp.\ disadvantaged, compared to the other in agents in~$\mathcal{M}$. The greedy covering procedure by \cite{farnadi2021individual}, for example, which identifies a subset $\Sol'$ of the optimal solutions such that each agent in $\M$ appears in at least one solution in~$\Sol'$, may falsely consider agents to belong to $\Y$ while they actually belong to~$\M$. Consider, for example, an instance with three agents and a set of optimal solutions~$\Sol = \{(1,1,0),(1,0,1),(0,1,1)\}$. While all agents belong to~$\M$, any greedy covering only contains two solutions in~$\Sol$, and will falsely consider one of the agents to belong to~$\Y$ instead of~$\M$.

The following proposition shows that this partitioning can be done by calling the solver at most $n + 1$ times for ILPs in~$\Xi$ that differ from the original formulation in at most one constraint, regardless of the size of~$\mathcal{S}$.

\begin{proposition}
\label{prop:partitioning_YNM}
Given an integer linear program~$\xi \in \Xi$, we can partition the set of agents~$\mathcal{A}$ into disjoint subsets~$\mathcal{Y}$, $\mathcal{N}$ and~$\mathcal{M}$ by solving at most $n+1$ integer linear programs in~$\Xi$ that differ in at most one constraint from~$\xi$.
\end{proposition}

\begin{proof}
First, we find an optimal solution $(\x^*,\y^*)$ for $\xi$. Next, for each agent $i\in\A$, we solve at most one ILP $\xi_i$ in $\Xi$ which is identical to $\xi$ with the additional constraint that $x_i = 1-x^*_i$. If the optimal objective value of $\xi_i$ is not equal to the optimal objective value of $\xi$, or if~$\xi_i$ is infeasible, then agent $i$ belongs to $\Y$ if $x_i^*=1$, and to $\N$ if $x^*_i=0$. If $\xi_i$ and $\xi$ have the same optimal objective value, however, then $i\in\M$. 

Note that careful bookkeeping can potentially reduce the number of ILPs to be solved. If an optimal solution has already been computed in which $x_j=1-x^*_j$ for an agent $j\in\A$, then $j\in\M$, and $\xi_j$ does not have to be optimized.
\end{proof}

In the remainder of this paper, we will assume that we know the partitioning of the set of agents~$\mathcal{A}$ into~$\mathcal{Y}$, $\mathcal{N}$, and~$\mathcal{M}$, unless stated otherwise. Moreover, because any lottery~$\bm{\lambda}\in\Delta(\mathcal{S})$ will always, resp.\ never, select the agents in~$\mathcal{Y}$, resp.~$\mathcal{N}$, we will only focus on the selection probabilities of the agents in~$\mathcal{M}$. Considering that any ILP~$\xi \in \Xi$ can be transformed into an equivalent ILP in which the agents in~$\mathcal{Y} \cup \mathcal{N}$ are replaced by parameters, we assume the set of agents~$\mathcal{A}$ to be equal to~$\mathcal{M}$ in the remainder of this paper, unless stated otherwise.


\section{Distribution rules}
\label{sec:distributions}
In this section, we will introduce several distribution rules and their computational properties. We propose frameworks to find distributions optimizing a linear or a concave objective function, and we illustrate the proposed frameworks to find distributions maximizing the egalitarian and the Nash social welfare. Note, however, that our frameworks can be easily modified to find distributions optimizing other objective functions. Moreover, we propose a method to apply the \emph{Random Serial Dictatorship} rule, which has been extensively studied in the social choice and matching literature, to our setting.

In general, we can distinguish two different ways to realize a distribution~$\bm{d}$ for a specific integer linear program in practice. First, one could explicitly find a lottery~$\bm{\lambda}$ that realizes~$\bm{d}$, together with the optimal solutions in~$\mathcal{S}$ with a strictly positive weight, and then select solution~$(\bm{x^s},\bm{y^s})\in \mathcal{S}$ with probability~$\lambda_s$. Secondly, one could specify a method that outputs only one solution~$(\bm{x^s},\bm{y^s})\in \mathcal{S}$ according to an underlying lottery~$\bm{\lambda}$ that realizes~$\bm{d}$, without explicitly generating all relevant solutions in~$\mathcal{S}$. The following definition formalizes this distinction \citep[a similar distinction for probabilistic assignments has been made by][]{demeulemeester2022pessimist}.

\begin{definition}
\label{def:representation_implementation} Given an integer linear program~$\xi \in \Xi$ and a distribution~$\bm{d} \in [0,1]^n$ that is realizable over the corresponding set of optimal solutions~$\mathcal{S}$,
\begin{enumerate}[label=(\roman*)]
\item a \emph{decomposition} of~$\bm{d}$ is a tuple~$(\mathcal{S}', \bm{\lambda}')$, with~$\mathcal{S}' \subseteq \mathcal{S}$ and~$\bm{\lambda'} \in \Delta_{\bm{d}}(\mathcal{S}')$;
\item an \emph{implementation} of~$\bm{d}$ is an algorithm that randomly selects a single solution~$(\bm{x^s},\bm{y^s})\in \mathcal{S}$ according to a lottery~$\bm{\lambda}\in\Delta_{\bm{d}}(\mathcal{S})$.
\end{enumerate}
\end{definition}

By Carathéodory's theorem, for any distribution there exists a decomposition in which at most $n+1$ optimal solutions have a strictly positive weight. Such a decomposition can be found by applying the algorithm described in Theorem 6.5.11 by \cite{grotschel2012geometric}.

Note that, given an implementation, neither the distribution which it realizes, nor the underlying decomposition from which a solution is sampled are assumed to be explicitly known (see Section~\ref{subsec:RSD} for an example). Clearly, a decomposition of a distribution~$\bm{d}$ implies its implementation, but not vice versa. Given the computational complexity of generating optimal solutions in integer programming, as discussed in Section~\ref{sec:related}, obtaining a decomposition of a distribution is not always tractable. We will therefore pay special attention to cases in which we can find an implementation of a distribution without first generating its decomposition (see Section~\ref{subsec:RSD}).

\subsection{Uniform}
\label{subsec:uniform}
The \emph{uniform} distribution is the distribution resulting from selecting each solution in $\Sol$ with equal probability.
\begin{definition}
\label{def:uniform}
Given an integer linear program~$\xi \in \Xi$, the \emph{uniform} distribution~$\bm{d^U}$ is the distribution realized by lottery~$\bm{\lambda^U} = \{\frac{1}{|\mathcal{S}|}, \ldots, \frac{1}{|\mathcal{S}|}\}$.
\end{definition}

We are not aware of any general method to obtain~$\bm{d^U}$ for all ILPs in $\Xi$ without counting all optimal solutions in~$\mathcal{S}$. Although doing so is~$\#\mathcal{P}$-complete in general, as discussed in Section~\ref{sec:related}, this approach might still be tractable in practice for smaller instances \citep[see, e.g.,][]{farnadi2021individual}.

The main advantage of the uniform rule is that the random selection of one of the optimal solutions allows for a transparent implementation. Nevertheless, linking the agents' selection probabilities directly to the number of optimal solutions in which they are selected might not be considered fair in many applications.\footnote{Consider the introductory example in Section~\ref{sec:intro}: does the mere fact that the twins ($x_1$) are selected in more of the optimal solutions than each of the individual students ($x_2, x_3, x_4$) justify the higher selection probability of the twins under the uniform rule?} \cite{farnadi2021individual} and \cite{flanigan2021transparent} study the application of the uniform rule for specific problem settings.

\subsection{Leximin distribution}
\label{subsec:leximin}
Next, we discuss a distribution rule that aims to determine the selection probabilities from an egalitarian perspective. Clearly, a distribution rule that selects each agent in $\M$ with the same probability is not realizable for all instances in $\Xi$, as is illustrated in Example~\ref{ex:intro}. Therefore, we focus on a distribution that is egalitarian in nature, and that will always be realizable, by construction, namely the \emph{leximin} distribution. The intuition behind the leximin distribution is to first maximize the lowest selection probability for the agents in~$\mathcal{M}$, then to maximize the second-lowest selection probability, etc. We propose an algorithm to compute the leximin distribution by iteratively generating optimal solutions to be used in its decomposition.

For any distribution~$\bm{d}\in \mathbb{R}^n$, denote by~$\text{lex}(\bm{d})\in\mathbb{R}^n$ the vector that is obtained by reordering the elements of~$\bm{d}$ in non-decreasing order. Given an ILP~$\xi\in \Xi$, we say that a distribution~$\bm{d}\in[0,1]^n$ \emph{lexicographically dominates} a distribution~$\bm{q}\in[0,1]^n$ when either~$\text{lex}(\bm{d})_1 > \text{lex}(\bm{q})_1$, or there exists an index~$i\in\{2, \ldots, n\}$ such that~$\text{lex}(\bm{d})_i > \text{lex}(\bm{q})_i$ while~$\text{lex}(\bm{d})_j = \text{lex}(\bm{q})_j$ for all~$1\leq j<i$.

\begin{definition}
    \label{def:leximin}
    Given an integer linear program~$\xi \in \Xi$, a \emph{leximin} distribution~$\bm{d^L}\in\text{Conv}(\mathcal{S}_{\bm{x}})$ is a distribution that is not lexicographically dominated by any other distribution~$\bm{q}\in\text{Conv}(\mathcal{S}_{\bm{x}})$.
\end{definition}

Note that there will be a unique leximin distribution for each ILP $\xi\in\Xi$. Imagine, by contradiction, that there would be two leximin distributions~$\bm{p}$ and $\bm{q}$. Then the average of~$\bm{p}$ and $\bm{q}$ would lexicographically dominate both~$\bm{p}$ and $\bm{q}$.

Unlike for the uniform distribution, it is possible to find a decomposition of~$\bm{d^L}$ without counting the number of solutions in~$\Sol$ by using a similar approach as \citet[][Theorem 1]{airiau2022portioning}. Each iteration of our algorithm consists of two steps, denoted as the upper and the lower problem. In the upper problem, we start by identifying the largest value such that all agents whose selection probabilities have not yet been fixed in the previous iterations can be selected with at least that probability. Next, in the lower problem, we identify the agents whose selection probabilities are exactly equal to this value in the leximin distribution, we fix their selection probabilities to the obtained value, and we proceed to the next iteration. Whereas \cite{airiau2022portioning} explicitly know the set of possible outcomes, however, we will adopt a \emph{column generation} approach in each step of the algorithm to avoid full enumeration of the optimal solutions. 

In each iteration~$t$ of the algorithm, let $N_t\subseteq \M$ denote the set of agents whose selection probabilities have been fixed in the previous iterations, where $N_t=\varnothing$ in the first iteration. First, we find the largest value $\gamma$ for which there still exists a distribution $\bm{d}\in\text{Conv}(\Sol_{\x})$ that selects all agents in $\M\setminus N_t$ with a probability of at least $\gamma$. We will do this by solving the column generation framework $\left[\text{RMP}_t\right]$, which corresponds to the upper problem of our algorithm in iteration~$t$. Consider an ILP~$\xi \in \Xi$, and denote by~$\tilde{\mathcal{S}} \subseteq \mathcal{S}$ the subset of the optimal solutions that we initially include in the restricted master problem. Then we set~$\gamma^*$ equal to the objective value of the following linear program $\left[\text{RMP}_t\right]$, where decision variable~$\lambda_s$ refers to the weight of solution~$(\bm{x^s},\bm{y^s})\in\tilde{\mathcal{S}}$ in the corresponding lottery, and where~$d^L_i$ refers to the selection probabilities that were fixed in the previous iterations for the agents in~$N_t$.

\begin{subequations}
	\begin{align}
	&\left[\text{RMP}_t\right] & &\max &\gamma &&\\ 
	& & &\: \text{s.t.} &\sum_{s=1}^{|\tilde{\mathcal{S}}|} \lambda_s x^s_i &\geq \gamma  &\forall \; i \in \mathcal{M}\setminus N_t,\label{con:rmp_1}\\
	& & & &\sum_{s=1}^{|\tilde{\mathcal{S}}|} \lambda_s x^s_i &= d^L_i  &\forall \; i \in N_{t},\label{con:rmp_4}\\
	& & & &\sum_{s=1}^{|\tilde{\mathcal{S}}|}  \lambda_s &= 1,  & \label{con:rmp_2}\\
	& & & &\lambda_{s} &\geq 0  &\forall \; s\in\{1, \ldots, |\tilde{\mathcal{S}}|\}. \label{con:rmp_3}
	\end{align}
\end{subequations}

Denote the dual variables related to constraints~(\ref{con:rmp_1}), (\ref{con:rmp_4}), and~(\ref{con:rmp_2}) by~$\bm{\mu}\in\mathbb{R}^{|\mathcal{M}\setminus N_t|}_-$, $\bm{\nu}\in\mathbb{R}^{|N_t|}$, and~$\rho\in\mathbb{R}$, respectively. Then an optimal solution for~$\left[\text{RMP}_t\right]$, with dual variables~$\bm{\mu^*},\bm{\nu^*}$ and~$\rho^*$, is optimal over all solutions in~$\mathcal{S}$ if no solution~$\bm{x^s}\in\mathcal{S}_{\bm{x}}$ has a positive \emph{reduced cost}. This means that for all solutions~$\bm{x}\in\mathcal{S}_{\bm{x}}$ the following should hold:

\begin{equation}
\label{eq:reduced_cost}
-\sum_{i \in \mathcal{M}\setminus N_t}\mu^*_ix_i -\sum_{i \in N_t}\nu^*_ix_i - \rho^* \leq 0.
\end{equation}

The pricing problem then consists of the constraints of the original problem~$\xi$, and an additional constraint to enforce that the original objective value is optimal, i.e.,~$\bm{v}^\intercal \bm{x} + \bm{w}^\intercal \bm{y} = z^*$, while the objective function of the pricing problem maximizes the left-hand side of Equation~(\ref{eq:reduced_cost}). If the pricing problem finds a solution~$\bm{x'}\in\mathcal{S}_{\bm{x}}\setminus\tilde{\mathcal{S}_{\bm{x}}}$ with a strictly positive reduced cost, then~$\bm{x'}$ is added to~$\tilde{\mathcal{S}_{\bm{x}}}$ and the restricted master problem $\left[\text{RMP}_{t}\right]$ is solved again. An optimal solution~$\gamma^*$ over~$\mathcal{S}$ is found in iteration~$t$ when the pricing problem cannot find a solution with a strictly positive objective value.

Next, after we have found~$\gamma^*$, we want to identify the agents in~$\M$ that are selected with a probability equal to~$\gamma^*$ in the leximin distribution by solving the lower problem. Note that simply fixing the probabilities for all agents for whom constraints (\ref{con:rmp_1}) are binding might result in a lexicographically dominated distribution, because some of these agents might be selected with a higher probability in the leximin distribution, while other agents might be selected with probability~$\gamma^*$ as well. To verify whether agent $j\in\M\setminus N_t$ is selected with probability~$\gamma^*$, we solve the following linear program over a subset $\tilde{\Sol}\subseteq\Sol$ of the optimal solutions:

\begin{subequations}
	\begin{align}
	&\left[\text{LP}_{jt}\right] & &\max &\theta &&\\ 
	& & &\: \text{s.t.} &\sum_{s=1}^{|\tilde{\mathcal{S}}|} \lambda_s x^s_j &\geq \gamma^*+\theta,  &\label{con:lp_1}\\
	& & & &\sum_{s=1}^{|\tilde{\mathcal{S}}|} \lambda_s x^s_i &\geq \gamma^*  &\forall \; i \in \mathcal{M}\setminus N_t,\label{con:lp_2}\\
	& & & &\sum_{s=1}^{|\tilde{\mathcal{S}}|} \lambda_s x^s_i &= d^L_i  &\forall \; i \in N_{t},\label{con:lp_3}\\
	& & & &\sum_{s=1}^{|\tilde{\mathcal{S}}|}  \lambda_s &= 1,  & \label{con:lp_4}\\
	& & & &\lambda_{s} &\geq 0  &\forall \; s\in\{1, \ldots, |\tilde{\mathcal{S}}|\}. \label{con:lp_5}
	\end{align}
\end{subequations}

Similarly to the column generation procedure for the upper problem $\left[\text{RMP}_t\right]$, denote the dual variables of constraints (\ref{con:lp_1})-(\ref{con:lp_4}) by $\pi\in\mathbb{R}_{-}$, $\bm{\mu}\in\mathbb{R}^{|\mathcal{M}\setminus N_t|}_-$, $\bm{\nu}\in\mathbb{R}^{|N_t|}$, and~$\rho\in\mathbb{R}$, respectively. Then an optimal solution for $\left[\text{LP}_{jt}\right]$ with dual variables $\pi^*$, $\bm{\mu^*}$, $\bm{\nu^*}$, and~$\rho^*$ is optimal over all solutions in~$\Sol$ if for all solutions $\x\in\Sol_{\x}$ it holds that

\begin{equation}
\label{eq:reduced_cost2}
-\pi^*x_j - \sum_{i \in \mathcal{M}\setminus N_t}\mu^*_ix_i -\sum_{i \in N_t}\nu^*_ix_i - \rho^* \leq 0.
\end{equation}

Hence, the pricing problem of formulation $\left[\text{LP}_{jt}\right]$ consists of maximizing the left-hand side of Equation~(\ref{eq:reduced_cost2}) over the constraints of the original ILP~$\xi$, and an additional constraint~$\bm{v}^\intercal \bm{x} + \bm{w}^\intercal \bm{y} = z^*$ to enforce that the original objective value is optimal. If the pricing problem finds a solution~$\bm{x'}\in\mathcal{S}_{\bm{x}}\setminus\tilde{\mathcal{S}_{\bm{x}}}$ with a strictly positive reduced cost, then~$\bm{x'}$ is added to~$\tilde{\mathcal{S}_{\bm{x}}}$ and the restricted master problem $\left[\text{LP}_{jt}\right]$ is solved again. An optimal solution~$\theta^*$ over~$\mathcal{S}$ is found in iteration~$t$ when the pricing problem cannot find a solution with a strictly positive objective value.

If $\theta^* = 0$ is the optimal objective value of $\left[\text{LP}_{jt}\right]$ over the set of all optimal solutions~$\Sol$ for an agent $j\in\M\setminus N_t$, then we add agent~$j$ to~$N_t$, and we set $d^L_j=\gamma^*$. Note that there must be at least one agent $j\in\M\setminus N_t$ for whom this is the case, because otherwise~$\gamma^*$ was not the optimal objective value of the upper problem $\left[\text{RMP}_t\right]$ over all solutions in~$\Sol$. When the lower problem $\left[\text{LP}_{jt}\right]$ has been solved for all agents $j \in \M\setminus N_t$, the algorithm proceeds to the next iteration unless $\M=N_t$.

Clearly, the described algorithms will require at most $|\mathcal{M}|$ iterations to output a decomposition $(\tilde{\Sol},\bm{\lambda})$ of~$\bm{d^L}$, where~$\bm{\lambda}$ are the weights that are found in the last linear program that was solved, and $\tilde{\Sol}\subseteq\Sol$ is the subset of the optimal solutions that has been generated throughout the algorithm. This implies that the upper problem $\left[\text{RMP}_t\right]$ should be solved at most $|\M|$ times to optimality over all solutions in~$\Sol$, and that formulation $\left[\text{LP}_{jt}\right]$ in the lower problem should be solved at most $\frac{|\M|}{2}(|\M|+1)$ times to optimality over all solutions in~$\Sol$. Generally, later iterations will be less computationally heavy, because the solution distribution from the previous iteration and the corresponding subset of optimal solutions can be used as a \enquote{warm start} by a solver.

We are not aware of any method to directly obtain an implementation of~$\bm{d^L}$ without first constructing one of its decompositions.

\subsection{Custom selection criteria}
\label{subsec:custom_selection_criteria}
Given an integer linear program~$\xi \in \Xi$, assume that a decision-maker wants to find a distribution~$\bm{d^f}$ that minimizes a convex function~$f:\mathbb{R}^{n}\to\mathbb{R}$, or, equivalently, that maximizes a concave function~$-f$. The choice of this function is problem-specific. One could, for example, minimize the $k$-norm 
\begin{equation}
    L_k(\bm{d}) = \left(\sum_{i\in\mathcal{M}}d^k_i\right)^{\frac{1}{k}},
\end{equation}
\noindent for a real number~$k\geq 1$ \citep{farnadi2021individual}. Alternatively, one could maximize the geometric mean

\begin{equation}
    f^N(\bm{d}) = \left(\prod_{i\in\mathcal{M}}d_i\right)^{\frac{1}{|\mathcal{M}|}}.
\end{equation}
\noindent While the solution that maximizes the geometric mean is also known as the \emph{Nash (bargaining) solution} \citep{nash1950bargaining} in the related literature on cooperative bargaining games, it is rather known as the \emph{maximum Nash welfare} solution in social choice literature \citep[e.g.,][]{caragiannis2019unreasonable}. 

Assuming the full set of optimal solutions~$\mathcal{S}$ for a given ILP $\xi\in\Xi$ is known, we can find a decomposition of~$\bm{d^f}$ using the following formulation $\left[\mathcal{C}_f(\Sol)\right]$, where the decision variables~$\bm{d^f}\in[0,1]^n$ and~$\bm{\lambda}\in\Delta_{\bm{d^f}}(\mathcal{S})$ represent a distribution and a realizing lottery, respectively, and each element $(\bm{x^s},\bm{y^s})\in\Sol$ corresponds to an optimal solution of~$\xi$.

\begin{subequations}
	\begin{align}
	& \left[\mathcal{C}_f(\Sol)\right]& &\min &f(\bm{d^f}) &&\label{con:general_obj}\\ 
	& & &\: \text{s.t.} &\sum_{s=1}^{|\mathcal{S}|} \lambda_s \bm{x^s} &= \bm{d^f},  &\label{con:general_1}\\
	& & & &\sum_{s=1}^{|\mathcal{S}|}  \lambda_s &= 1,  & \label{con:general_2}\\
	& & & &\lambda_{s} &\geq 0  &\forall \; s\in\{1, \ldots, |\mathcal{S}|\}. \label{con:general_3}
	\end{align}
\end{subequations}

When the set of optimal solutions~$\mathcal{S}$ is not known and cannot be fully enumerated efficiently, however, the form of the objective function~$f$ plays a crucial role. For a linear objective function~$f$, such as the arithmetic mean, a \enquote{classical} column generation approach as described for formulation $\left[\text{RMP}_t\right]$ in Section~\ref{subsec:leximin} can be adopted in a straightforward way. For a non-linear objective function~$f$, however, this approach is no longer possible, and we will discuss how to adapt the column generation procedure for convex programs that satisfy strong duality. 

A similar approach has been recently proposed in Section 8 of the Supplementary information by \cite{flanigan2021fair} for sortition, which is a special case of the ILPs in class~$\Xi$. \cite{zangwill1967convex} propose a simplex-type algorithm that allows for delayed variable generation, but our particular problem setting allows for a simpler optimality condition than his general setting. While our framework is similar in spirit to \emph{simplical decomposition} \citep[e.g.,][]{holloway1974extension, von1977simplicial}, as it iteratively uses the gradient to identify an improving solution, our framework maintains the original non-linear objective function in the restricted master problem, while simplical decomposition methods rely on its linear inner approximation. \cite{chicoisne2023computational} provide a recent overview on column generation methods for non-linear optimization problems.

Denote by $[\mathcal{C}_f(\tilde{\Sol})]$ the variant of formulation $[\mathcal{C}_f(\Sol)]$ which minimizes a differentiable and convex function~$f$ over a subset $\tilde{\Sol}\subseteq\Sol$ of the optimal solutions. Denote the dual variables related to constraints (\ref{con:general_1})-(\ref{con:general_3}) in $[\mathcal{C}_f(\tilde{\Sol})]$ by $\bm{\mu}\in\mathbb{R}^{|\M|}$, $\rho\in\mathbb{R}$, and $\bm{\nu}\in\mathbb{R}^{|\tilde{\Sol}|}$. The column generation procedure proceeds as follows. In each iteration~$t$, we solve $[\mathcal{C}_f(\tilde{\Sol}^t)]$, where $\tilde{\Sol}^{1}$ is some non-empty subset of~$\Sol$ in the first iteration, and~$\tilde{\Sol}^t$ is defined in the previous iteration, otherwise. Let $(\bm{\lambda^*},\bm{d^*})$ denote an optimal solution to the primal problem $[\mathcal{C}_f(\tilde{\Sol}^t)]$ with dual variables $\bm{\mu^*}$, $\rho^*$ and $\bm{\nu^*}$, and let $(\bm{x'},\bm{y'})$ be an optimal solution to the convex program $\min\{\sum_{i\in\M}\mu^*_ix_i:(\bm{x},\bm{y})\in\Sol\}$. As we discuss in detail in Appendix~\ref{app:proof_of_correctness}, program $[\mathcal{C}_f(\tilde{\Sol}^t)]$ satisfies strong duality, and the Karush-Kuhn-Tucker conditions therefore imply that in its primal optimum $\mu_i^*=\frac{\partial}{\partial d_{i}} f(\bm{d^*})$ should hold. We have found a distribution~$\bm{d^*}$ which minimizes~$f$ over the solutions in~$\Sol$, and a corresponding decomposition $(\tilde{\Sol}^t,\bm{\lambda^*})$ of $\bm{d^*}$, if the following optimality condition holds: 
\begin{equation}
    \label{eq:optimality_condition}
    \sum_{i\in\M}\mu^*_ix^{old}_i\leq\sum_{i\in\M}\mu^*_ix'_i,
\end{equation}
where $(\bm{x^{old}},\bm{y^{old}})\in\tilde{\Sol}^t$ is any solution with $\lambda_{old}^*>0$. Note that the left-hand side of Condition~(\ref{eq:optimality_condition}) has a constant value of $-\rho^*$. If Condition~(\ref{eq:optimality_condition}) does not hold, we let $\tilde{\Sol}^{t+1} = \tilde{\Sol}^{t}\cup \{(\bm{x'},\bm{y'})\}$, and proceed to the next iteration.

In Appendix~\ref{app:proof_of_correctness}, we prove the correctness of this procedure by showing that it terminates after a finite number of iterations, and by showing that Condition~(\ref{eq:optimality_condition}) indeed implies the optimality of a distribution~$\bm{d^*}$.
    
\subsection{Random Serial Dictatorship}
\label{subsec:RSD}
Informally speaking, the \emph{Random Serial Dictatorship} (RSD) distribution~$\bm{d^{RSD}}$ is the expected outcome of the following procedure. After randomly ordering the agents, the first agent in this order selects the solutions in~$\mathcal{S}$ in which she is selected, then the second agent selects the solutions in which she is selected \emph{among the remaining solutions}, etc. The procedure ends when a unique solution remains, which will occur by construction.

Define $\bm{\sigma}=\left(\sigma(1),\ldots,\sigma(|\M|)\right)$ to be a strict ordering over the agents in~$\mathcal{M}$, and denote the set of all orderings by~$\Sigma$. Moreover, consider the \emph{Serial Dictatorship} function~$SD:\Xi\times\Sigma\to\mathcal{S}$, which, given an ILP~$\xi \in \Xi$ and an ordering~$\bm{\sigma}\in\Sigma$, will output one of the solutions in~$\mathcal{S}$ according to the procedure described above. Our definition of SD coincides with the common definition in voting \citep[e.g.,][]{aziz2014parametrized}. Using this notation, we can define the RSD distribution as follows.

\begin{definition}
\label{def:RSD}
    Given an integer linear program~$\xi \in\Xi$, the \emph{Random Serial Dictatorship} (RSD) distribution~$\bm{d^{RSD}}\in[0,1]^n$ is given by
    \begin{equation}
        \bm{d^{RSD}} = \frac{1}{|\mathcal{M}|!}\sum_{\bm{\sigma} \in \Sigma} SD(\xi,\sigma).
    \end{equation}
\end{definition}

Obtaining a decomposition of~$\bm{d^{RSD}}$ is not straightforward. First, one should find all optimal solutions in~$\mathcal{S}$, because they represent the alternatives from which the agents can choose. Second, even given the set of optimal solutions, it is~$\#\mathcal{P}$-complete to determine the exact probabilities in the RSD distribution \citep{aziz2013computational}. Only when each of the agents in~$\mathcal{M}$ is selected in exactly one of the optimal solutions in~$\mathcal{S}$, the RSD distribution can be calculated in linear time \citep{aziz2013computational}.

We observe that, similarly to the assignment and the voting setting where the RSD mechanism is well-studied, computing the exact RSD probabilities is computationally challenging, whereas implementing its result is rather straightforward. In fact, we will discuss two different implementations of~$\bm{d^{RSD}}$, which both modify the formulation to enforce the random ordering of the agents in~$\mathcal{M}$.

First, assume that the objective weights~$\bm{v}$ and $\bm{w}$ of some ILP $\xi\in\Xi$ are integer. In that case, given a random ordering~$\bm{\sigma}\in\Sigma$ of the agents, one can then simply perturb the objective function of~$\xi$. Denote the order of agent~$i\in\mathcal{M}$ in~$\bm{\sigma}$ by~$\sigma^{-1}(i)$. Moreover, define the perturbation vector~$\bm{\delta} = \left(\frac{1}{2^k}\right)_{k=1}^{|\mathcal{M}|}$. When we replace the objective function of~$\xi$ by 

\begin{equation}
\label{eq:objective_perturbation}
\sum_{i\in\mathcal{A}\setminus\mathcal{M}}v_ix_i + \sum_{i\in \mathcal{M}}\left(v_i + \delta_{\sigma^{-1}(i)}\right)x_i + \bm{w}^\intercal\bm{y},
\end{equation}
\noindent each solution~$(\bm{x^s},\bm{y^s})\in\mathcal{S}$ will be found according to an underlying lottery~$\bm{\lambda^{RSD}}$ that realizes~$\bm{d^{RSD}}$. To show that perturbation~$\bm{\delta}$ obtains the desired result, note that any perturbation~$\bm{\delta'}\in\mathbb{R}^{|\mathcal{M}|}$ will implement~$\bm{d^{RSD}}$ if the two following requirements are satisfied.
\begin{enumerate}[label=(\roman*)]
\item The obtained solution after the perturbation should still be optimal to the original problem. Because of our assumption that the objective coefficients~$\bm{v}$ and~$\bm{w}$, and decision variables~$\bm{x}$ and~$\bm{y}$ are integer, the difference between the objective values of an optimal and a non-optimal solution is greater than or equal to one. Hence, a perturbation~$\bm{\delta'}$ should satisfy~$\sum_{i=1}^{|\mathcal{M}|}\delta'_i < 1$.
\item The order of the agents in the random ordering~$\bm{\sigma}\in\Sigma$ should be respected. Given integer $\bm{v}, \bm{w}, \bm{x},$ and $\bm{y}$, a sufficient condition for this requirement to hold is that a perturbation~$\bm{\delta'}$ satisfies $\delta'_i > \sum_{j:j>i}\delta'_j$.
\end{enumerate}
Clearly, perturbation~$\bm{\delta}$ satisfies both requirements. A clear advantage of using objective perturbation~$\bm{\delta}$ is that we can implement~$\bm{d^{RSD}}$ by solving one ILP in $\Xi$ that only differs from the original formulation by its perturbed objective function. Indeed, although we assumed that the partition of the set of agents into~$\mathcal{Y}$,~$\mathcal{N}$, and~$\mathcal{M}$ is known, it is also possible to extend~$\bm{\delta}$ such that it contains a value for each of the agents in~$\mathcal{A}$, and to then solve~$\xi$ with the perturbed objective function for some random ordering of the agents in~$\mathcal{A}$. A drawback of using perturbation~$\bm{\delta}$ is that numerical issues could occur for a large number of agents. More specifically, the precision of the solver might not be able to distinguish between the agents that appear at the end of the random ordering~$\bm{\sigma}$. In order to circumvent this issue for a solver with precision~$\omega$, one can choose to perturb only the objective coefficients of at most the first $\lfloor-\log_2(\omega)\rfloor$ agents in~$\bm{\sigma}$, fix their solution values, and to then do the same for the next $\lfloor-\log_2(\omega)\rfloor$ agents, etc.

A second method to implement~$\bm{d^{RSD}}$ neither depends on the precision of the solver, nor requires integer objective coefficients~$\bm{v}$, but iteratively solves at most~$|\mathcal{M}|$ integer linear programs in~$\Xi$ using Algorithm~$\ref{alg:implement_RSD}$. Algorithm~$\ref{alg:implement_RSD}$ will first find an optimal solution in which the first-ranked agent in~$\bm{\sigma}$ is selected. Next, the algorithm will verify whether there exists an optimal solution in which the two first-ranked agents in~$\bm{\sigma}$ are selected. If such a solution exists, we enforce that the second-ranked agent is selected in the remainder of the algorithm. The algorithm continues until all agents in $\bm{\sigma}$ have been checked in this manner.

\begin{algorithm}[hbt!]
 \caption{Iterative implementation of~$\bm{d^{RSD}}$}
    \label{alg:implement_RSD}
    \hspace*{\algorithmicindent} \textbf{Input:} $\xi \in \Xi, \bm{\sigma}\in\Sigma,$ $z^*$ \\
    \hspace*{\algorithmicindent} \textbf{Output:} $(\bm{x^s},\bm{y^s})\in\mathcal{S}$
    \begin{algorithmic}[1]
    \For{$i\in\{1,\ldots,|\mathcal{M}|\}$}
    \State $\xi \gets \xi$ with additional constraint~$x_{\sigma(i)} = 1$
    \State Obtain an optimal solution~$(\bm{x},\bm{y})$ with objective value~$z$ for~$\xi$
    \If{$z = z^*$}
    \State $(\bm{x^s},\bm{y^s})\gets(\bm{x},\bm{y})$
    \Else 
    \State $\xi \gets \xi$ without constraint~$x_{\sigma(i)} = 1$
    \EndIf
    \EndFor
    \end{algorithmic}
\end{algorithm}


\section{Axiomatic implications}
\label{sec:axioms}
In this section, we study which axiomatic properties are satisfied by the distribution rules described in Section~\ref{sec:distributions}. Interestingly, the following result implies that all axiomatic results that have been obtained for collective choice under dichotomous preferences \citep{bogomolnaia2004random,bogomolnaia2005collective}, also hold for distribution rules over optimal solutions of integer linear programs in~$\Xi$.

\begin{proposition}
\label{prop:solution_space}
For every possible set of outcomes $\mathcal{T}\subseteq2^\A$, there exists an ILP $\xi\in\Xi$ such that $\Sol_{\x}(\xi) = \mathcal{T}$, where $\Sol_{\x}(\xi)$ is the projection of~$\Sol(\xi)$ on the $\x$-variables. 
\end{proposition}
\begin{proof}
Consider an arbitrary set of possible outcomes $\T = \{\bm{o^1},\ldots,\bm{o^{|\T|}}\}$. Let $\xi_{\T}\in\Xi$ be an ILP with decision variables $\x\in\{0,1\}^n$ and a constant objective function. Additionally, let $\xi_{\T}$ have one \emph{no-good cut} for each of the outcomes in the complement of~$\T$, i.e.,
\begin{align}
    &\sum_{i\in\A:o^t_i = 0}x_i + \sum_{i\in\A:o^t_i = 1}(1-x_i) \geq 1 &\forall\; \bm{o^t}\in2^{\A}\setminus\T
\end{align}
Each of these constraints will forbid exactly one outcome that is not in~$\T$. As a result, it holds that $\Sol_{\x}(\xi_{\T}) = \T$.
\end{proof}

\begin{corollary}
\label{cor:axiomatic}
    All axiomatic results that have been obtained for collective choice under dichotomous preferences (fair mixing) also hold for distribution rules over optimal solutions of integer linear programs in $\Xi$.
\end{corollary}

Corollary \ref{cor:axiomatic} follows from the observation that there are no imposed constraints on the set of possible outcomes in collective choice under dichotomous preferences. Note that Proposition~\ref{prop:solution_space} is necessary for this result in order to show that the considered class of ILPs $\Xi$ is sufficiently rich, i.e., that its set of optimal solutions can be equal to any subset of $2^{\A}$. For relevant axiomatic properties in probabilistic social choice under dichotomous preferences, we refer the reader to \cite{bogomolnaia2004random}, \cite{bogomolnaia2005collective}, \cite{duddy2015fair}, \cite{aziz2019fair}, and \cite{brandl2021distribution}. 


When studying a specific class of problems that can be modeled by an ILP in~$\Xi$, such as kidney exchange or knapsack, it may be the case that there exist sets of outcomes that do not correspond to the set of optimal solutions of any instance of that specific problem class. To illustrate this, one can observe, for example, that the set $\{(1,0),(1,1)\}$ cannot correspond to the optimal solutions of any knapsack instance with two items in which the weights of the items in the objective function are strictly positive. Regardless, Corollary \ref{cor:axiomatic} still has the following implications with respect to the validity of the axiomatic results from collective choice under dichotomous preferences for such a specific subproblem that can be modeled by an ILP in~$\Xi$. First, all \emph{positive} results of the type \enquote{(rule) satisfies (axiom)} remain valid. Second, the \emph{negative} results of the type \enquote{(rule) does not satisfy (axiom)} are not guaranteed to hold for specific subproblems. To prove such negative axiomatic results for a specific class of subproblems, it suffices to provide an example instance in which a rule violates the considered axiom. Third, as a result, \emph{characterization} results of the type \enquote{(rule) is the only rule satisfying (set of axioms)} are also not guaranteed to hold for specific classes of subproblems that can be modeled by an ILP in~$\Xi$.

We will briefly discuss the axiomatic properties of the introduced distribution rules, but we refer the reader to \cite{aziz2019fair} for a complete overview.\footnote{Note that the rules we introduced are named differently in \cite{aziz2019fair}: RSD is referred to as \emph{Random Priority} (RP), leximin as \emph{Egalitarian}, and maximum Nash welfare as \emph{Nash max product} (NMP).} We will not discuss results with respect to strategy-proofness, because the type of information that is reported by the agents in fair integer programming depends on the application at hand (and will influence the constraints or the objective function). In any case, they do not simply report which of the outcomes they like, as is the case in collective choice under dichotomous preferences.

Table~\ref{tab:axioms} summarizes which axioms are satisfied by the discussed distribution rules. We say that a distribution rule is \emph{anonymous} if it treats agents symmetrically, i.e., the selection probabilities of the agents do not change when their names or labels are changed. Similarly, a distribution rule is \emph{neutral} if it treats outcomes symmetrically.

Furthermore, one could consider several proportionality axioms, which build on the idea that individuals and groups of like-minded agents should receive their \enquote{fair share} of the selection probabilities. From an individual perspective, the \emph{individual fair share} (IFS) property entails that each agent in $\M$ has at least a $\frac{1}{|\M|}$-fraction of the decision power, and is therefore selected with a probability of at least $\frac{1}{|\M|}$. Alternatively, given a subset $K\subseteq \M$ of agents who have identical preferences, i.e., who are selected in the same subset of the optimal solutions, the \emph{unanimous fair share} (UFS) property requires that each agent in $K$ is selected with a probability of at least $\frac{|K|}{|\M|}$, which is proportional to the size of that group. Clearly, UFS implies IFS. Lastly, \cite{aziz2019fair} propose two strengthenings of UFS, which impose bounds on the selection probabilities for groups of agents who are selected in the same optimal solution (\emph{average fair share}), or for coalitions of agents (\emph{core fair share}). We refer the reader to their paper for an exact definition of both properties.

Table~\ref{tab:axioms} shows that the Nash rule satisfies all of the introduced proportionality properties, and therefore provides the best guarantees to groups of agents, while RSD only satisfies UFS, and leximin only satisfies IFS. The uniform rule, which is not discussed in \cite{aziz2019fair}, even violates IFS for $|\M|\geq 3$. To illustrate this, one can observe that for a set of optimal solutions $\{(1,0,0), (0,1,0), (0,0,1), (0,1,1)\}$, the uniform rule would select the first agent with a probability of $\frac{1}{4}<\frac{1}{|\M|}=\frac{1}{3}$. Table~\ref{tab:axioms} also shows that any distribution rule that deterministically selects one of the optimal solutions, which is the approach that is currently adopted by solvers such as Gurobi \citep{gurobi2023deterministic} or CPLEX \citep{cplex2021deterministic}, violates all of the introduced proportionality properties, including the weakest axioms of anonymity and neutrality.

Lastly, a feasible distribution $\bm{d}\in\text{Conv}(\Sol_{\x})$ is \emph{Pareto-efficient} if there is no alternative distribution $\bm{d'}\in\text{Conv}(\Sol_{\x})$ such that $\bm{d'}\geq\bm{d}$, and at least one inequality $d'_i\geq d_i$ is strict. As shown by \cite{aziz2019fair}, the leximin and the Nash rules output a Pareto-efficient distribution, whereas RSD and the uniform rule violate this axiom. While the deterministic rule is Pareto-efficient when all $\x$-variables have strictly positive objective coefficients, the way in which the deterministic rule selects an optimal solution determines whether it is still Pareto-efficient when some of the objective coefficients are zero.

\begin{table}[!htb]
    \centering
    \small
    \hspace*{-5cm}
    \caption{Axiomatic properties of the discussed rules, where \enquote{determ.} refers to returning one of the optimal solutions in a deterministic way \citep{aziz2019fair}.}
    \begin{tabular}{lccccc}
    \toprule
         & \textbf{determ.\ } & \textbf{uniform} & \textbf{leximin} & \textbf{RSD} & \textbf{Nash} \\
         \midrule
         
         \textbf{Anonymity \& neutrality} & \no & \yes & \yes &\yes &\yes \\
         \textbf{Individual fair share}& \no & \no & \yes& \yes& \yes\\
         \textbf{Unanimous fair share} & \no & \no& \no& \yes& \yes\\
         \textbf{Average fair share} & \no & \no& \no& \no& \yes\\
         \textbf{Core fair share} & \no &\no& \no& \no &\yes\\
         \midrule
         \textbf{Pareto-efficiency} & - & \no & \yes &\no & \yes \\
         \bottomrule
    \end{tabular}
    \label{tab:axioms}
\end{table}


\section{Cardinal preferences\label{sec:extensions}}
In this section, we study how to extend the proposed methods to problems where agents have cardinal preferences over the solutions, i.e., they associate a real value to each of the optimal solutions. 

\label{app:cardinal}
\subsection{Notation}
Consider the larger class of mixed-integer linear programs $\Theta$ which is equal to~$\Xi$ except for the fact that $\x\in\mathbb{R}^n$ instead of $\x\in\{0,1\}^n$. We then let the utility that an agent experiences when an optimal solution for a formulation in $\Theta$ is selected be equal to the value of the agent's $\x$-variable in that solution. We assume that the agents' utilities satisfy the Von Neumann-Morgenstern axioms, allowing us to compare lotteries over the optimal solutions. We additionally make the assumption that the convex hull formed by the optimal solutions of a formulation in $\Theta$ is bounded, which follows from the assumption that the agents do not experience infinite utility in any of the optimal solutions.

Instead of partitioning the agents into the sets $\Y, \M$, and $\N$ (Proposition~\ref{prop:partitioning_YNM}), we are now interested in the highest and the lowest utilities they experience from any of the optimal solutions. We define the \emph{utopia point} $\bm{u}(\xi)$ of a formulation $\xi \in \Theta$ as the point in which each of the agents receive their maximally attainable utility in any of the optimal solutions, i.e., $u_i(\xi) = \max\{x_i: (\x,\y) \in \Sol(\xi)\}$, where $\Sol(\xi)$ denotes the set of optimal solutions of~$\xi$.\footnote{The utopia point is also referred to as the \emph{ideal point} or the \emph{aspiration point} in the literature on cooperative bargaining.} Similarly, we define the \emph{dystopia point}~$\bm{o}(\xi)$ as the point in which the utility of each of the agents is equal to the lowest utility they receive in any of the optimal solutions of $\xi$, i.e., $o_i(\xi) = \min\{x_i: (\x,\y) \in \Sol(\xi)\}$. Clearly, both the utopia and the dystopia point can be found by solving $n$ modified versions of the original formulation, each minimizing/maximizing the utility of one specific agent under the additional constraint that the original objective value is equal to the optimum objective value of~$\xi$. Let~$\M$ refer to the subset of the agents for whom the values in the dystopia and the utopia point differ.

\subsection{Connection with cooperative bargaining}
\label{subsubsec:bargaining}
This generalized setting is closely related to the $n$-person cooperative bargaining problem, as the feasible region in the $n$-person bargaining problem is also generally assumed to be non-empty, convex, closed, and bounded \citep[e.g.,][]{thomson1994cooperative,peters2013axiomatic}. There are two main differences with our setting, however. First, the literature on cooperative bargaining mostly focuses on the axiomatic properties of the solution concepts, and less on computing or implementing the resulting lotteries.

Second, the general assumption in cooperative bargaining is the existence of a \emph{disagreement point} in the feasible region. This point will be the selected outcome if the agents fail to reach a unanimous agreement. Additionally, it is generally assumed that there exists another point in the feasible region that Pareto-dominates the disagreement point. In our setting, however, it is not clear which point to select as the disagreement point. In fact, a Pareto-dominated disagreement point that belongs to the feasible region might not even exist for some fair integer programming instances. Consider, for example, the linear program $\max \{x_1+x_2: x_1+x_2=1, x_1, x_2\ge0\}$, where the feasible region is the line between optimal solutions $(1,0)$ and $(0,1)$.


Nevertheless, most solution concepts from cooperative bargaining do not crucially depend on the belonging of the disagreement outcome to the feasible region. For this reason, and because of the ambiguity in choosing a disagreement point in our setting, we propose to replace the role of the disagreement point by the dystopia point $\bm{o}(\xi)$.

\subsection{Distribution rules and axiomatic results}
Whereas we measured fairness of a solution by simply comparing the selection probabilities of the agents in the case of dichotomous preferences, comparing the unscaled utilities of the agents could lead to extremely unbalanced solutions for cardinal preferences. Instead, we will compare the agents' utilities to the dystopia or the utopia point, as this reflects how much the utility of an agent changes compared to their worst or best utility in any of the optimal solutions. By replacing the disagreement point with the dystopia point, the column generation frameworks from Section \ref{subsec:leximin} (for linear objective functions) and \ref{subsec:custom_selection_criteria} (for minimizing convex objective functions) can be used to find lotteries over the set of optimal solutions representing several well-known solution concepts from cooperative bargaining.

In the spirit of the leximin rule (Section~\ref{subsec:leximin}), \cite{raiffa1953arbitration} and \cite{kalai1975other} studied the solution concept that maximizes the fraction of the maximum possible utility improvement of the worst-off agent, with respect to the disagreement point. When replacing the disagreement point by the dystopia point, this is equivalent to finding a distribution~$\bm{d}$ that maximizes the following objective function: 
\begin{equation}
\label{eq:card_maximin}
    \max\left\{\min_{i\in \M}\frac{d_i-o_i(\xi)}{u_i(\xi) - o_i(\xi)}: \bm{d} \in \text{Conv}(\Sol(\xi))\right\}.
\end{equation}
\cite{imai1983individual} extended the Raiffa-Kalai-Smorodinsky solution by lexicographically maximizing the vector containing the fractions of the maximum possible utility improvement that are experienced by the agents in the resulting distribution. Both solution concepts can be implemented using the column generation framework from Section~\ref{subsec:leximin} by modifying Constraints~(\ref{con:rmp_1}) accordingly, and by replacing the objective function~(\ref{eq:reduced_cost}) to be maximized in the pricing problem by
\begin{equation}
\label{eq:reduced_cost_cardinal}
-\sum_{i \in \mathcal{M}\setminus N_t}\frac{\mu^*_ix_i}{u_i(\xi)-o_i(\xi)} -\sum_{i \in N_t}\nu^*_ix_i - \rho^*.
\end{equation}

The Nash rule from Section~\ref{subsec:custom_selection_criteria} was originally introduced by \cite{nash1950bargaining} for the two-person bargaining game. In the case of cardinal preferences, it is the distribution that maximizes the product of the differences between an agent's utility in the solution and their utility in the dystopia point:
\begin{equation}
    \max\left\{\prod_{i\in \M}\left(d_i-o_i(\xi)\right): \bm{d} \in \text{Conv}(\Sol(\xi))\right\}.
\end{equation}

The RSD rule from Section~\ref{subsec:RSD} can also be extended by letting the agents sequentially retain the optimal solutions that maximize their utility difference between the final distribution and the dystopia point. While an iterative approach similar to Algorithm~\ref{alg:implement_RSD} will still work, the first implementation of RSD described in Section~\ref{subsec:RSD}, which relies on perturbing the objective function, is no longer applicable.

Many other solution concepts have been proposed in the literature on $n$-person cooperative bargaining that could be implemented in our setting using the column generation procedures from Sections \ref{subsec:leximin}-\ref{subsec:custom_selection_criteria}, and we believe it is an interesting research direction to identify attractive rules for our setting. We refer the reader to \cite{thomson2022axiomatic} for a survey on recent results.

Lastly, many cooperative bargaining solution concepts have been axiomatically characterized \citep{thomson1994cooperative,thomson2022axiomatic,peters2013axiomatic}. While our setting differs slightly because of the difficulty of identifying a disagreement point within the feasible region, replacing it by the dystopia point does not affect the validity for most of the discussed axioms. A detailed study of which axiomatic results can, and cannot, be judiciously transferred to fair integer programming lies outside the scope of our paper, however.

\section{Near-optimal solutions\label{sec:near_opt}}
Almost all of the solution methods that were described in Sections \ref{sec:partitioning}, \ref{sec:distributions}, and \ref{sec:extensions} can be extended in a straightforward way to find distributions over optimal and near-optimal solutions, regardless of whether the agents have dichotomous or cardinal preferences. Only the first implementation of RSD, which was based on perturbing the objective function, is no longer valid when including near-optimal solutions.

Let $\Sol^{\epsilon}(\xi)$ denote the set of solutions whose objective values are at most a fraction $\epsilon$ worse than the optimal objective value $z^*(\xi)$, i.e., $\bm{v^{\intercal}}\bm{x^j} + \bm{w^{\intercal} y^j} \geq (1-\epsilon)z^*(\xi)$ for all solutions $(\bm{x^j},\bm{y^j})\in \Sol^{\epsilon}(\xi)$. For the partitioning algorithm (Proposition \ref{prop:partitioning_YNM}) and for the iterative implementation of RSD (Algorithm~\ref{alg:implement_RSD}), finding a distribution over the optimal and near-optimal solutions simply implies checking whether a solution belongs to $\Sol^\epsilon(\xi)$ instead of to $\Sol(\xi)$. When imposing that a solution belongs to $\Sol^\epsilon(\xi)$ in the other solution methods, we can simply replace the constraint that the objective value is equal to $z^*(\xi)$ by $\bm{v^{\intercal}}\bm{x^j} + \bm{w^{\intercal} y^j} \geq (1-\epsilon)z^*(\xi)$. From an axiomatic point of view, the discussion in Section~\ref{sec:axioms} remains unchanged.

\section{Computational experiments}
\label{sec:comput}
In this section, we investigate the performance of the distribution rules that were discussed in Section~\ref{sec:distributions}. We study the proposed distribution rules for two distinct problems: kidney exchange, where the agents have dichotomous preferences, and the single-machine scheduling problem to minimize total tardiness, where the agents have cardinal preferences. We evaluate how the proposed distributions compare to the optimal Nash product and to the optimal minimum selection probability, and we compare the required computation times to obtain them. 

We compare the exact methods that were introduced in Section~\ref{sec:distributions} with the following two heuristics:
\begin{enumerate}[label = (\roman*)]
    \item \textbf{Re-index}: change the order in which the $\x$-variables are entered into the solver according to a random ordering of the agents $\bm{\sigma}$,
    \item \textbf{Perturb}: perturb the objective coefficients of each agent~$i$ with a small value $\gamma_i$ that is generated from the uniform distribution $[-\frac{1}{n}\sum_{i\in\A}v_i,$ $\frac{1}{n}\sum_{i\in\A}v_i]$. 

\end{enumerate}


\subsection{Computational setup}
\label{subsec:comp_setup}
Before describing our findings, we first discuss the details of the implementation, and the evaluated formulation for the kidney exchange problem and the problem of minimizing the total tardiness on a single machine.

\subsubsection{Implementation details}
\label{subsec:impl_details}
All experiments are implemented with \CC, compiled with Microsoft Visual Studio 2019, and run on an AMD Ryzen 7 PRO 3700U processor running at 2.30 GHz, with 32GB of RAM memory on a Windows 10 64-bit OS. All linear and integer linear programs are solved using Gurobi 10.0, with default parameter settings, and with a precision of~$10^{-5}$ to avoid numerical issues.

In the implementation of the algorithm to find a partitioning of the agents into sets $\Y$, $\M$, and $\N$ (Proposition \ref{prop:partitioning_YNM}), we add a callback to the solver that aborts the optimization as soon as the best upper bound on the objective function is smaller than the known optimal objective value, because we are only interested in knowing whether or not an optimal solution exists with the inclusion of an additional constraint in each step.

In the implementation of the column generation framework for the leximin rule (Section~\ref{subsec:leximin}), formulations $\left[\text{RMP}_t\right]$ and $\left[\text{LP}_{jt}\right]$ are modeled using a single model by changing the objective function, and by adding and removing constraint~(\ref{con:lp_1}) when required. Additionally, both frameworks are initiated with a subset $\tilde{\Sol}$ of the optimal solutions such that each of the agents in~$\M$ is selected in at least one of the solutions in~$\tilde{\Sol}$. Such a subset is found using a greedy algorithm, which iteratively adds a constraint to the original formulation to enforce the selection of at least one of the agents who is not yet selected by the solutions in $\tilde{\Sol}$, until the model becomes infeasible.

Lastly, for the distribution rules that adopt randomness (RSD, perturb and re-index heuristics) or that require all optimal solutions (uniform) we limit the number of iterations/found solutions to 1,000. As such, their reported performances are only approximated. 

The code to run these experiments, as well as the evaluated instances, are available online (\url{https://github.com/DemeulemeesterT/Fair-Integer-Programming.git}).

\subsubsection{Kidney exchange}
\label{subsec:Kidney}
The first application we consider is the kidney exchange problem, in which (incompatible) patient-donor pairs are matched with each other in such a way that the matched donors' kidneys can be successfully transplanted. The kidney exchange problem is known to be $\mathcal{NP}$-hard when the maximum allowed length of the exchange cycles is at least equal to three \citep{abraham2007clearing}. We implement the cycle formulation for this problem \citep{abraham2007clearing,roth2007efficient} because of its clear intuition. Note, however, that while more efficient formulations exist, the scope of this paper is not to find the most efficient formulation for a problem, but simply to assess the performance of the discussed distributions for a given formulation.

Let $V$ denote the set of all patient-donor pairs, and let $\mathcal{C}$ denote the set of all cycles of such pairs such that the donor of a pair in the cycle is compatible with the patient of the next pair in the cycle. Let $x_v$ be a binary decision variable which equals one if the patient from pair $v\in V$ receives a transplant. Moreover, let $y_c$ be a binary decision variable which equals one if cycle $c\in\mathcal{C}$ is selected for an exchange. Consider the following formulation $\left[\text{IP}_{\text{KE}}\right]$, which maximizes the number of executed transplants:
\begin{subequations}
\begin{align}
		& \left[\text{IP}_{\text{KE}}\right]& &\max &\sum_{v\in V} x_v&\\
		& & &\text{s.t.} \;\; &\sum_{c\in\mathcal{C}:v\in c}y_c&=x_v&\forall \;v\in V,\label{con:KE1}\\
        & & & &x_v &\in \{0,1\} &\forall\; v\in V,\label{con:KE2_1}\\
        & & & &y_c &\in \{0,1\} &\forall\; c\in\mathcal{C}.\label{con:KE2_2}
\end{align} 
\end{subequations}

We evaluate formulation $\left[\text{IP}_{\text{KE}}\right]$ on the kidney exchange instances that were used in \cite{farnadi2021individual}, in which the number of patient-donor pairs ranges from 10 to 70, with 50 instances for each size. The data are based on the US population characteristics presented in \cite{saidman2006increasing}, and were generated using the generator proposed by \cite{constantino2013new}. We only consider cycles of length at most three, following the observation by \cite{roth2007efficient} that cycles of size four or larger can often be decomposed into cycles of size at most three. Our methods can be extended to larger cycles in a straightforward way, as well as to formulations which allow for transplant chains initiated by altruistic donors.

\subsubsection{Minimize total tardiness on a single machine}
\label{subsubsec:tardiness}
The second application for which we evaluate the proposed rules is the classic problem of scheduling jobs on a single machine in order to minimize the total tardiness. \cite{du1990minimizing} showed that this problem is weakly $\mathcal{NP}$-hard, with a pseudo-polynomial algorithm provided by \cite{lawler1977pseudopolynomial}.

Let $N = \{1,\ldots,n\}$ denote a set of jobs, where each job $j\in N$ has processing time~$p_j\in\mathbb{N}$, and due date~$d_j\in\mathbb{N}$. Denoting the completion time of job~$j$ by~$C_j$, the \emph{tardiness}~$T_j$ of that job is defined as $T_j = \max\{C_j-d_j,0\}$. We assume that the corresponding agent of job~$j$ receives a utility~$u_j$ that is equal to the negative tardiness, i.e., $u_j = -T_j$. This implies that agents are indifferent about how much before the due date their job is completed, but their utility decreases linearly with each additional time period that their job is scheduled beyond the due date.

Among the many existing formulations for this problem, we implement the formulation with time index variables, as described by \cite{keha2009mixed}, because they find that it performs best for the related problem of minimizing the total \textit{weighted} tardiness on a single machine. Let $T = \sum_j p_j$ denote the latest possible completion time of any job. We define a binary decision variable $y_{jt}$ which is equal to one if job $j\in N$ is scheduled at time $t=0,\ldots,T-1$, and to zero otherwise. Moreover, let decision variable $x_j\leq 0$ denote the utility which the agent corresponding to job $j\in N$ receives from the solution. Consider the following formulation $\left[\text{IP}_{\text{TT}}\right]$ to find a schedule with minimal total tardiness, which is equivalent to maximal utilitarian utility: \newpage
\begin{subequations}
\begin{align}
		& \left[\text{IP}_{\text{TT}}\right]& &\max &\sum_{j\in N} x_j&\\
		& & &\text{s.t.}  \;\; &\sum_{t=0}^{T-1} y_{jt} &= 1&\forall\;j\in N,\label{con:TT1}\\
        & & & &\sum_{j=1}^n \sum_{s=\max\{0,t-p_j+1\}}^t y_{js} &\leq 1 &\forall \; t=0,\ldots,T-1,\label{con:TT2}\\
        & & & &\sum_{t=0}^{T-1}\left(t\cdot y_{jt}\right) + p_j &\leq d_j - x_j &\forall \; j \in N, \label{con:TT3}\\
        & & & &x_j \leq 0, \; y_{jt} &\in \{0,1\} &\forall\; j\in N \notag \\
        \label{con:TT5} & & & & & &t = 0,\ldots,T-1.
\end{align} 
\end{subequations}
Constraints~(\ref{con:TT1}) impose that each job should be scheduled exactly once, whereas Constraints~(\ref{con:TT2}) impose that only one job can be processed at each moment in time. Lastly, Constraints~(\ref{con:TT3}) describe the utilities of the agents.

We generate problem instances as described in \cite{chu1992branch} and \cite{baptiste2004branch}. Processing times~$p_j$ are uniformly sampled between 1 and 10. The due dates~$d_j$ are uniformly sampled from the interval $[0,\beta\sum_jp_j]$, where $\beta$ is a parameter. The number of jobs~$n$ in the generated instances ranges from 10 to 25, with step size five, and we evaluate three values of~$\beta$. For each combination of~$n$ and~$\beta$, we generate 50 instances.

\subsection{Results}
\label{subsec:results}
Figure~\ref{fig:MinNashKE} illustrates how the different distributions perform with respect to the minimum selection probability and the Nash product of the agents in $\M$, compared to the optimum. For each value of $|\M|$, the results in Figure~\ref{fig:MinNashKE} are averaged over all instances with that number of agents in~$\M$. In general, we can conclude that most distributions perform worse on the criterion that they do not optimize as instances grow larger.

For both problems, the performance of the leximin and the Nash distributions are close to the optimum on the criterion that they do not optimize. The performance of the other distribution rules, however, differs for the two studied applications.

\begin{figure}[t]
\includegraphics[width=\linewidth]{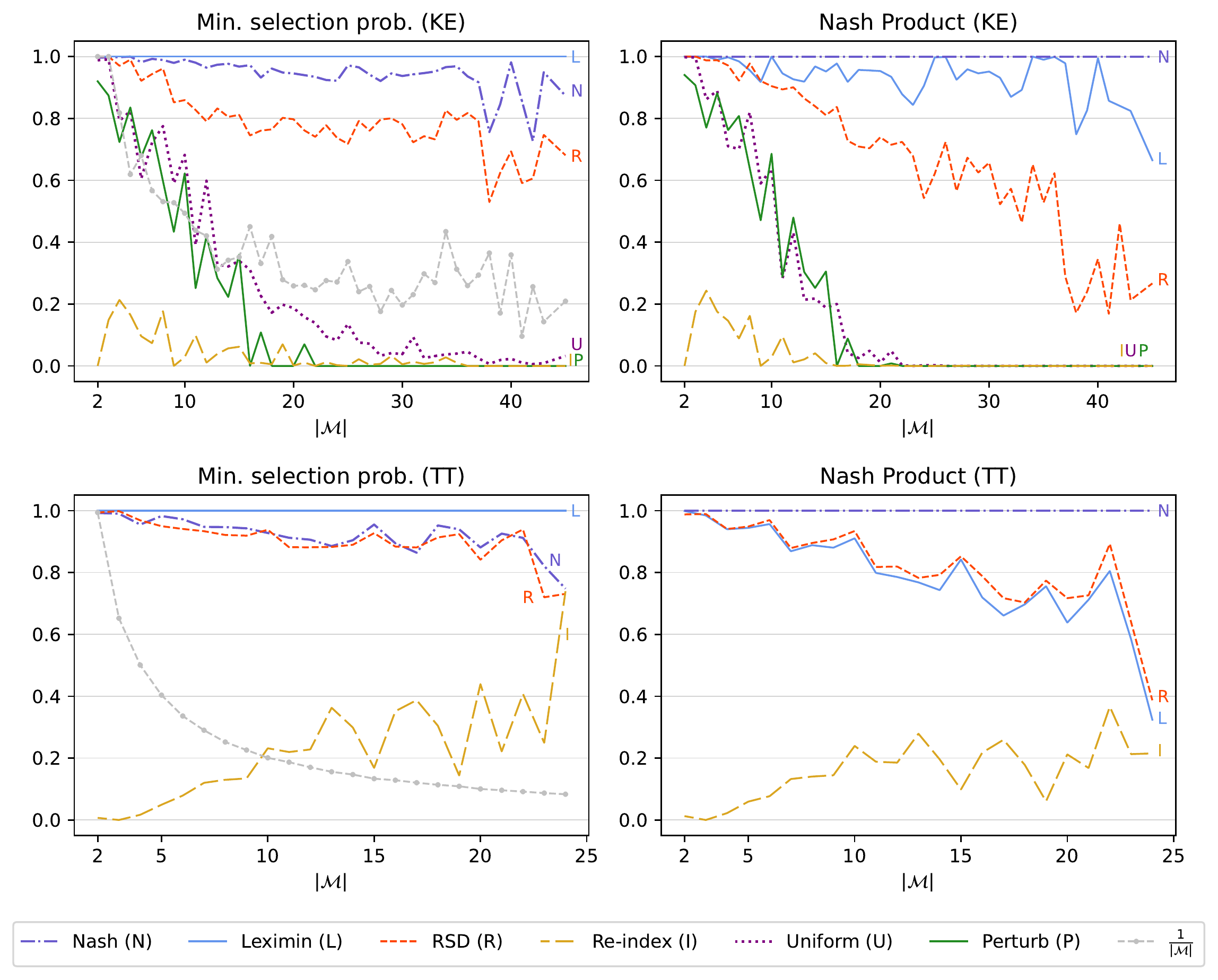}
\caption{Average ratio of minimum selection probability (left) and Nash product (right) for agents in $\M$ compared to the optimum in kidney exchange (KE - above) and total tardiness (TT - below) instances, with respect to the number of agents in $\M$.}
\label{fig:MinNashKE}
\end{figure}

For kidney exchange, the RSD distribution obtains the third-best performance on all criteria. Moreover, the uniform distribution, as well as the perturb and the re-index heuristics, all have decreasing ratios that are close to or equal to zero for larger instances. The minimum selection probability by the uniform distribution, for example, was less than 5\% of the optimum in 30 out of the 50 instances with 70 patient-donor pairs, and for the perturb and re-index heuristics, the minimum was less than 5\% in almost all instances of size 70. In comparison, the minimum selection probability by RSD is at least 40\% of the optimum in all kidney exchange instances.

\begin{table*}[t]
\caption{CPU time (in s) to find an optimal solution ($t_{\text{opt}}$), to find a partitioning of the agents into $\Y$, $\M$, and $\N$, and to compute the different distributions for the kidney exchange (KE) and the total tardiness (TT) instances. $|\M|$ denotes the average number of agents in $\M$. An asterisk (*) indicates that the distribution requires a prior partitioning of the agents, of which the time is not included.\label{tab:timeKE}}
\centering
\footnotesize
\begin{tabular}{crrrrrrr}
\toprule
\textbf{inst} & $|\M|$ & $\mathbf{t}_{\text{opt}}$ & \textbf{Partition} &\textbf{RSD*} & \textbf{Leximin*}&\textbf{Nash*}&\textbf{Uniform}\\
\midrule
KE10 & 1.5 & 0.002 & 0.009 & 0.004  & 0.008 & 0.125  & 0.006  \\
KE20 & 5.5 & 0.005 & 0.064 & 0.005  & 0.043 & 0.166  & 2.046   \\
KE30 & 9.3 & 0.009 & 0.158 & 0.008  & 0.124 & 0.267  & 16.069  \\
KE40 & 14.0 & 0.014 & 0.305 & 0.014  & 0.329 & 0.590  & 48.750  \\
KE50 & 17.4 & 0.020 & 0.469 & 0.020  & 0.577 & 0.967  & 77.588  \\
KE60 & 27.5 & 0.026 & 0.714 & 0.033  & 1.337 & 3.301  & 83.760  \\
KE70 & 29.5 & 0.029 & 1.046 & 0.054  & 1.699 & 14.425 & 115.546 \\
\midrule
TT10-0.05 & 5.6 & 0.033 & 0.596 & 0.129 & 0.221 & 0.196 & \ensuremath{-} \\
TT10-0.25 & 4.3 & 0.045 & 0.715 & 0.076 & 0.262 & 0.244 & \ensuremath{-} \\
TT10-0.50 & 3.4 & 0.058 & 0.878 & 0.024 & 0.268 & 0.247 & \ensuremath{-} \\
TT15-0.05 & 10.4 & 0.061 & 1.764 & 0.484 & 0.771 & 2.482 & \ensuremath{-} \\
TT15-0.25 & 7.9 & 0.115 & 2.720 & 0.438 & 1.371 & 1.112 & \ensuremath{-} \\
TT15-0.50 & 4.9 & 0.161 & 3.659 & 0.250 & 1.175 & 0.551 & \ensuremath{-} \\
TT20-0.05 & 15.1 & 0.103 & 3.975 & 1.158 & 1.996 & 70.927 & \ensuremath{-} \\
TT20-0.25 & 11.8 & 0.377 & 10.758 & 1.245 & 5.429 & 13.040 & \ensuremath{-} \\
TT20-0.50 & 8.0 & 0.544 & 23.589 & 1.592 & 6.769 & 4.267 & \ensuremath{-} \\
TT25-0.05 & 20.4 & 0.262 & 11.131 & 2.548 & 6.567 & 142.255 & \ensuremath{-} \\
TT25-0.25 & 15.8 & 0.626 & 22.527 & 2.494 & 11.585 & 24.382 & \ensuremath{-} \\
TT25-0.50 & 10.0 & 1.053 & 79.118 & 5.143 & 20.301 & 30.638 & \ensuremath{-}\\
\bottomrule
\end{tabular}
\end{table*}

For the problem of minimizing total tardiness, however, we observe that the RSD rule performs equally well as leximin and Nash on the criterion that they are not optimizing. Interestingly, we also observe that the ratios of the re-index heuristic are increasing for both criteria as instances grow larger, in contrast to the poor performance of the re-index heuristic for kidney exchange. For the instances with the most agents in $\M$, the ratios of the re-index heuristic are similar to those of the other rules that do not optimize that criterion. The good performance of the re-index heuristic might be caused by the fact that cardinal preferences allow for more \enquote{intermediate} integer solutions that lie in between the agents' dystopia and utopia points, in contrast to dichotomous preferences, but further research is required to determine whether this is a general or an application-specific observation.

For reference, the left panels of Figure~\ref{fig:MinNashKE} also show how the ratio of $\frac{1}{|\M|}$ performs compared to the discussed rules, reflecting a situation in which all agents in $\M$ have an equal share of the decision power. We showed in Section~\ref{sec:axioms} that the uniform rule violates the individual fair share property, and our computational experiments for kidney exchange show that this is not merely a theoretical result, but a common observation in practice.

Table~\ref{tab:timeKE} displays the required computational effort to find a partitioning of the agents, and to obtain each of the distributions. The computation time for finding an implementation of RSD for dichotomous preferences is, as expected, close to that of finding a single optimal solution, excluding the time to find a partitioning of the agents into $\Y$, $\M$, and $\N$. This is particularly true for kidney exchange, as we can apply the first implementation of RSD, which perturbs the objective function, whereas we apply the iterative implementation (Algorithm~\ref{alg:implement_RSD}) for minimizing total tardiness. Combining this observation with the performance of RSD in Figure \ref{fig:MinNashKE}, this presents RSD-variants as a pragmatic method to control the selection probabilities of the optimal solutions. Furthermore, we observe that, overall, the computation times for the leximin distribution scale better than for the Nash distribution for the evaluated instances, with the Nash rule having a particularly high variance in solution times (see Appendix~\ref{app:stdev}).

\section{Conclusion and future research directions}
\label{sec:conclusion}
Fair integer programming studies how to control the selection probabilities of the optimal solutions of integer linear programs with binary, integer or real decision variables. Our computational experiments show that rules such as the Random Serial Dictatorship (RSD) mechanism manage to combine a reasonably good performance on the evaluated welfare criteria with similar computation times to those of finding a single optimal solution. This result illustrates that addressing fair integer programming does not necessarily cause an increase in the computation times. Given the prevalence of integer programming formulations having multiple optimal solutions, we believe, therefore, that controlling the selection probabilities of the optimal solutions should be an essential step for decision-makers and practitioners when making high-impact decisions using integer programming techniques.

We identify three major directions for future research. First, our paper focuses on developing general-purpose algorithms that can be applied to a wide class of integer linear programs. The design of dedicated algorithms for specific problems that exploit the combinatorial structure of the problem at hand is an interesting research direction. Efficient algorithms to compute the maximin distribution have been proposed, for example, by \cite{li2014egalitarian} for kidney exchange, or by \cite{garcia2020fair} when the optimal solutions form a matroid. An interesting analysis of the combinatorial structure of fair distribution rules is the recent work by \cite{hojny2023fairness}. 

Second, we suggest investigating the existence of distributions other than RSD that can be implemented in similar computation times to those of finding a single optimal solution, possibly inspired by the wide range of solution concepts in cooperative bargaining, or by introducing fairness considerations into the literature on symmetry breaking in integer programming. 

Lastly, we have compared our column generation procedures to rules of which the underlying distributions could only be accurately computed by generating many optimal solutions (uniform, RSD, perturb, re-index). Therefore, we could only work with relatively small instances. A more extensive computational study of the column generation frameworks and of RSD implementations on larger instances and on other types of problems is a promising direction for future work.

\setlength{\parskip}{10pt}
\footnotesize\noindent \textbf{Acknowledgements} Tom Demeulemeester is funded by PhD fellowship 11J8721N of Research Foundation - Flanders. We would like to thank Markus Brill, Ágnes Cseh, Jannik Matuschke, and the anonymous reviewers for their valuable comments and suggestions.
\setlength{\parskip}{0pt}

\normalsize

\bibliographystyle{elsarticle-harv} 
\bibliography{bibl.bib}






\appendix
\section{Proof of correctness column generation procedure Section~\texorpdfstring{\ref{subsec:custom_selection_criteria}}{4.3}}
\label{app:proof_of_correctness}
We adopt the same notation as in Section~\ref{subsec:custom_selection_criteria}. Before proving the correctness of the column generation procedure described in Section~\ref{subsec:custom_selection_criteria}, we point out that the following two conditions are satisfied. First, for any non-empty subset of the optimal solutions $\tilde{\Sol}$, convex program $[\mathcal{C}_f(\tilde{\Sol})]$ is clearly feasible. Second, for any convex function~$f:\mathbb{R}^n\to\mathbb{R}$ and for any non-empty $\tilde{\Sol}\subseteq\Sol$, convex program $[\mathcal{C}_f(\tilde{\Sol})]$ satisfies strong duality. This follows from the observation that $[\mathcal{C}_f(\tilde{\Sol})]$ satisfies Slater's condition \citep[e.g.,][Ch.\ 5.2.3]{boyd2004convex} for any non-empty subset $\tilde{\Sol}\subseteq\Sol$, because its constraints are linear. As a result, the Karush-Kuhn-Tucker (KKT) conditions for a solution $(\bm{\lambda},\bm{d^f})$ of $[\mathcal{C}_f(\tilde{\Sol})]$ to be optimal over all solutions in $\tilde{\Sol}$ are both necessary and sufficient.\footnote{Note that a variant of the described column generation would still work when additional, possibly convex, constraints are added to $[\mathcal{C}_f(\tilde{\Sol})]$, as long as the resulting convex program still satisfies strong duality \citep[e.g.,][]{flanigan2021fair}.}

Given a solution $(\bm{\lambda^*},\bm{d^*})$ of $[\mathcal{C}_f(\tilde{\Sol})]$ with dual variables $\bm{\mu^*}$, $\rho^*$ and $\bm{\nu^*}$, the KKT conditions imply that~$f$ is minimized over all optimal solutions in~$\tilde{\Sol}\subseteq\Sol$ in point $(\bm{\lambda^*},\bm{d^*})$ if and only if the following conditions are satisfied:

\begin{subequations}
\begin{align}
     \text{Constraints } &(\ref{con:general_1})-(\ref{con:general_3}),&\notag\\
     \nu^*_s &\geq 0 & \forall\;s=1,\ldots,|\tilde{\Sol}|,\label{con:KKT2}\\
     \nu^*_s\lambda^*_s&=0&\forall\;s=1,\ldots,|\tilde{\Sol}|,\label{con:KKT3}\\
     \frac{\partial}{\partial d_{i}} f(\bm{d^*}) &= \mu^*_i &\forall \;i\in\M,\label{con:KKT4}\\
     \sum_{i\in\M}\mu^*_ix^s_i + \rho^* &= \nu^*_s &\forall\;s=1,\ldots,|\tilde{\Sol}|.\label{con:KKT5}
\end{align}
\end{subequations}

First, we show that the proposed column generation procedure terminates after a finite number of iterations. To do so, it suffices to show that the new solution $(\bm{x'},\bm{y'})$ that is found in each iteration~$t$ was not yet in~$\tilde{S}^t$. Starting from Condition~(\ref{con:KKT5}), and by using the complementary slackness conditions~(\ref{con:KKT3}), the assumption that $\lambda^*_{old}>0$, and conditions~(\ref{con:KKT2}) that $\nu_s\geq0$ for all solutions in~$\tilde{\Sol}^t$, we obtain that
\begin{equation}
    \sum_{i\in\M}\mu^*_ix^{old}_i=-\rho^*=\sum_{i\in\M}\mu^*_ix^{s}_i-\nu_{s}^{*}\leq\sum_{i\in\M}\mu^*_ix^{s}_i,
\end{equation}
for all solutions $(\bm{x^{s}},\bm{y^{s}})\in\tilde{\Sol}^t$. When Condition~(\ref{eq:optimality_condition}) does not hold, this implies that $\sum_{i\in\M}\mu^*_ix'_i<\sum_{i\in\M}\mu^*_ix^{old}_i\leq \sum_{i\in\M}\mu^*_ix^{s}_i$, for all solutions $(\bm{x^{s}},\bm{y^{s}})\in\tilde{\Sol}^t$, which shows that $(\bm{x'},\bm{y'})\notin{\tilde{S}^t}$.

Second, we show that Condition~(\ref{eq:optimality_condition}) is indeed an optimality condition, and that a solution $(\bm{\lambda}^*,\bm{d}^*)$ of $[\mathcal{C}_f(\tilde{\Sol}^t)]$ minimizes~$f$ over all solutions in~$\Sol$ if Condition~(\ref{eq:optimality_condition}) holds. To show this, we will extend the variables~$\bm{\lambda^*}$ and~$\bm{\nu^*}$ to all solutions in $\Sol$. Let~$\lambda^*_s$ and $\nu^*_s$ remain unchanged for solutions $(\bm{x^s},\bm{y^s})\in\tilde{\Sol}^t$, and let $\lambda^*_r=0$ and $\nu^*_r = \sum_{i\in\M}\mu^*_ix^r_i + \rho^*$ for all solutions $(\bm{x^r},\bm{y^r})\in\Sol\setminus\tilde{\Sol}^t$. Next, we show that if Condition~(\ref{eq:optimality_condition}) holds, $\bm{d^*}$, $\bm{\mu^*}$, $\rho^*$, and these extended variables $\bm{\lambda^*}$, and $\bm{\nu^*}$ satisfy the KKT conditions for $[\mathcal{C}_f(\Sol)]$ over all solutions in~$\Sol$, and will therefore minimize~$f$ over~$\Sol$. 

Most of the conditions are directly implied by the fact that the KKT conditions hold for $[\mathcal{C}_f(\tilde{\Sol}^t)]$, because all variables in the equations remain the same (Equation~(\ref{con:KKT4}) remains unchanged, and Equations (\ref{con:general_3}), (\ref{con:KKT2}), (\ref{con:KKT3}), and (\ref{con:KKT5}) remain unchanged for all solutions $(\bm{x^s},\bm{y^s})\in\tilde{S}^t$). Clearly, the conditions imposed by Equations (\ref{con:general_1}) and (\ref{con:general_2}) are satisfied, because the newly introduced variables are equal to zero. Moreover, $\lambda_s = 0$ for all $(\bm{x^s},\bm{y^s})\in\Sol\setminus\tilde{\Sol}^t$, which satisfies the conditions imposed by Equations~(\ref{con:general_3}) and~(\ref{con:KKT3}). The conditions imposed by Equations~(\ref{con:KKT2}) are satisfied for all solutions $(\bm{x^s},\bm{y^s})\in\Sol\setminus\tilde{\Sol}^t$ by using our assumption that Condition~(\ref{eq:optimality_condition}) holds:
\begin{equation}
\label{eq:KKT_proof}
    \sum_{i\in\M}\mu^*_ix^s_i \geq \sum_{i\in\M}\mu^*_ix'_i \geq
    \sum_{i\in\M}\mu^*_ix^{old}_i \geq
    \sum_{i\in\M}\mu^*_ix^{old}_i - \nu^*_{old} = -\rho^*
\end{equation}
By definition of $\nu^*_s$ for all solutions $(\bm{x^s},\bm{y^s})\in\Sol\setminus\tilde{\Sol}^t$, Equation~(\ref{eq:KKT_proof}) implies that $\nu_s^*\geq 0$. Lastly, the definition of the new $\nu^*_s$ implies that conditions~(\ref{con:KKT5}) hold.

\section{Standard deviations of computation times}\label{app:stdev}

\begin{table}[H]
\caption{Standard deviations of the CPU time (in s) to find an optimal solution ($t_{\text{opt}}$), to find a partitioning of the agents into $\Y$, $\M$, and $\N$, and to compute the different distributions for the kidney exchange (KE) and the total tardiness (TT) instances. $|\M|$ denotes the average number of agents in $\M$. An asterisk (*) indicates that the distribution requires a prior partitioning of the agents, of which the time is not included.}\label{tab:std}
\centering
\footnotesize
\begin{tabular}{crrrrrrr}
\toprule
\textbf{inst} & $|\M|$ & $\mathbf{t}_{\text{opt}}$ & \textbf{Partition} &\textbf{RSD*} & \textbf{Leximin*}&\textbf{Nash*}&\textbf{Uniform}\\
\midrule
KE10 & 1.5 & 0.001 & 0.007 & 0.003 & 0.013 & 0.074 & 0.009 \\
KE20 & 5.5 & 0.005 & 0.055 & 0.004 & 0.061 & 0.148 & 4.966 \\
KE30 & 9.3 & 0.009 & 0.101 & 0.005 & 0.150 & 0.190 & 20.508 \\
KE40 & 14.0 & 0.014 & 0.145 & 0.007 & 0.341 & 0.536 & 36.124 \\
KE50 & 17.4 & 0.020 & 0.201 & 0.008 & 0.511 & 1.114 & 30.778 \\
KE60 & 27.5 & 0.026 & 0.289 & 0.016 & 1.305 & 4.255 & 22.324 \\
KE70 & 29.5 & 0.029 & 0.386 & 0.028 & 1.229 & 40.800 & 33.162 \\ \midrule
TT10-0.05 & 5.6 & 0.008 & 0.121 & 0.079 & 0.095 & 0.172 & \ensuremath{-} \\
TT10-0.25 & 4.3 & 0.027 & 0.312 & 0.099 & 0.264 & 0.218 & \ensuremath{-} \\
TT10-0.50 & 3.4 & 0.028 & 0.263 & 0.060 & 0.125 & 0.176 & \ensuremath{-} \\
TT15-0.05 & 10.4 & 0.039 & 0.450 & 0.129 & 0.311 & 6.523 & \ensuremath{-} \\
TT15-0.25 & 7.9 & 0.093 & 1.794 & 0.237 & 1.587 & 1.313 & \ensuremath{-} \\
TT15-0.50 & 4.9 & 0.112 & 3.945 & 0.301 & 1.109 & 0.492 & \ensuremath{-} \\
TT20-0.05 & 15.1 & 0.052 & 1.024 & 0.316 & 0.715 & 310.703 & \ensuremath{-} \\
TT20-0.25 & 11.8 & 0.364 & 10.806 & 0.616 & 6.884 & 22.678 & \ensuremath{-} \\
TT20-0.50 & 8.0 & 0.484 & 56.285 & 1.890 & 7.075 & 7.994 & \ensuremath{-} \\
TT25-0.05 & 20.4 & 0.208 & 4.554 & 0.576 & 3.593 & 325.369 & \ensuremath{-} \\
TT25-0.25 & 15.8 & 0.486 & 15.208 & 0.953 & 9.351 & 22.126 & \ensuremath{-} \\
TT25-0.50 & 10.0 & 0.928 & 176.583 & 12.269 & 40.300 & 52.123 & \ensuremath{-}\\ \bottomrule
\end{tabular}
\end{table}

\end{document}